\documentclass[11pt]{article}
\usepackage{amsmath}
\usepackage{amsfonts,amsthm}
\usepackage{latexsym}
\usepackage{amscd}
\usepackage{amssymb}
\usepackage{graphicx}
\usepackage{float}
\usepackage[applemac]{inputenc}
\usepackage{dsfont}
\usepackage{multirow}
\usepackage{dcolumn}
\usepackage{rccol}
\usepackage{enumerate}
\usepackage{scalefnt}
\usepackage{booktabs}
\usepackage{subfigure}
\usepackage{graphicx}
\usepackage{color}
\usepackage[colorlinks]{hyperref}

\usepackage{textcomp}
\usepackage{array}
\usepackage{supertabular}
\usepackage{hhline}

\makeatletter
\newcommand\arraybslash{\let\\\@arraycr}
\makeatother
\newcolumntype{+}{>{\global\let\currentrowstyle\relax}}
\newcolumntype{^}{>{\currentrowstyle}}

\newlength{\bracewidth}

\def\fudge{\mathchoice{}{}{\mkern.5mu}{\mkern.8mu}}
\def\bbc#1#2{{\rm \mkern#2mu\vbar\mkern-#2mu#1}}
\def\bbb#1{{\rm I\mkern-3.5mu #1}}
\def\bba#1#2{{\rm #1\mkern-#2mu\fudge #1}}
\def\bb#1{{\count4=`#1 \advance\count4by-64 \ifcase\count4\or\bba A{11.5}\or
   \bbb B\or\bbc C{5}\or\bbb D\or\bbb E\or\bbb F \or\bbc G{5}\or\bbb H\or
   \bbb I\or\bbc J{3}\or\bbb K\or\bbb L \or\bbb M\or\bbb N\or\bbc O{5} \or
   \bbb P\or\bbc Q{5}\or\bbb R\or\bbc S{4.2}\or\bba T{10.5}\or\bbc U{5}\or
   \bba V{12}\or\bba W{16.5}\or\bba X{11}\or\bba Y{11.7}\or\bba Z{7.5}\fi}}

\restylefloat{float}

\newtheorem{theorem}{Theorem}[section]

\newtheorem{lemma}{Lemma}[section]

\title{On the Dynamics of Dengue Virus type 2 with Residence Times and Vertical Transmission }
\author{Derdei Bichara$^1$, Susan A. Holechek$^{1,2}$, Jorg\'e Vel\'azquez-Castro$^3$, \\
Anarina L. Murillo$^1$ and Carlos Castillo-Chavez$^1$\\
$^1$ Simon. A. Levin Mathematical, Computational and Modeling Sciences Center,\\
Arizona State University, P.O. Box 873901, Tempe, AZ 85287-3901; \\
 $^2$ Center for Infectious Diseases and Vaccinology, The Biodesign Institute, \\
Arizona State University, P.O. Box 875401, Tempe, AZ 85287-5401\\
$^3$ Facultad de Ciencias F\'\i sico Matem\'aticas, Universidad Aut\'onoma de Puebla, \\
Apartado Postal 1152, 72001, Puebla, Pue., Mexico
}

\date{December 31, 2015}

\begin{document}
\maketitle

\begin{abstract}
A two-patch mathematical model of Dengue virus type 2 (DENV-2) that  accounts for vectors' vertical transmission and  between patches human dispersal is introduced. Dispersal is modeled via a Lagrangian approach. A host-patch residence-times basic reproduction number is derived and conditions under which the disease dies out or persists are established. Analytical and numerical results highlight the role of hosts' dispersal in mitigating or exacerbating disease dynamics. The framework is used to explore  dengue dynamics using, as a starting point, the 2002 outbreak  in the state of Colima, Mexico.%
\end{abstract}



\section{Introduction}

Dengue, a re-emerging vector-borne disease, is caused by members of the genus \textit{Flavivirus} in the family \textit{Flaviviridae} with four active antigenically distinct serotypes, DENV-1, DENV-2, DENV-3, and DENV-4 \cite{deubel1988nucleotide}. The pathogenicity of dengue can range from asymptomatic, mild dengue fever (DF), to dengue hemorrhagic fever (DHF), and dengue shock syndrome (DSS) \cite{deubel1988nucleotide,halstead2002dengue}. Although infection with a dengue serotype does not usually protect against other serotypes, it is belief that secondary infections with a heterologous serotype increase the probability of DHF and DSS \cite{burke1988prospective,halstead1970observations}. According to the World Health Organization, 40\% of the global population is at risk for dengue infection with an estimate of 50 to 100 million infections yearly including 500,000 cases of DHF. It has been estimated that about 22,000 deaths, mostly children under 15 years of age, can be attributed to DHF \cite{WHO2009FactSheet}. In the United States, approximately 5\% or more of the Key West population in Florida was exposed to dengue during the 2009-2010 outbreak \cite{CDC2010} while the Hawaii Department of Health reported 190 cases during the 2015 outbreak on Oahu, the first outbreak since 2011. Since dengue is not endemic in Hawaii, health authorities have suggested that the recent outbreak may have been started by infected visitors  \cite{Hawaii2015}. Dengue is highly prevalent and endemic in Southeast Asia, which has experienced a 70\% increase in cases since 2004 \cite{kwok2010}; Mexico, also an endemic country,  reported during the 2002 outbreak over a million cases of DF and more than 17,000 cases of DHF\cite{guzman2003dengue,morens2008dengue}. Dengue is transmitted primarily by the vector \textit{Ae. aegypti}, which is now found in most countries in the tropics and sub-tropics \cite{harris2000clinical,reiter1997surveillance}. The secondary vector, \textit{Ae. albopictus}, has a range reaching farther north than \textit{Ae. aegypti} with eggs better adapted  to subfreezing temperatures \cite{hawley1987aedes,morens2008dengue}. Differences in susceptibility and transmission of dengue infection \cite{arunachalam2008natural,knox2003enhanced,tewari2004dengue} raise the possibility that some serotypes are either more successful at invading a host population, or  more pathogenic, or both\cite{kyle2008global}. DENV-2 is the most associated with dengue outbreaks involving DHF and DSS cases \cite{montoya-circulation,rico1997origins,sittisombut1997possible,zhang2006structure}, followed by DENV-1 and DENV-3 viruses \cite{balmaseda2006serotype,harris2000clinical,montoya-circulation}. While infection with any of the four dengue serotypes could lead to DHF, the rapid displacement of DENV-2 American  by DENV-2 Asian genotype has been linked to major outbreaks with DHF cases in Cuba, Jamaica, Venezuela, Colombia, Brazil, Peru and Mexico \cite{zhang2006structure,rico1997origins,sittisombut1997possible,montoya-circulation,lewis1993phylogenetic,rico1998molecular}. A possible mechanism involved in the dispersal and persistence of DENV-2 in nature is vertical transmission (transovarial transmission)  via \textit{Ae. aegypti}. Prior studies were unsuccessful in demonstrating vertical transmission via \textit{Ae. aegypti} \cite{rodhain1997mosquito}. However, the use of advances in molecular biology has shown that vertical transmission involving \textit{Ae. aegypti} and \textit{Ae. albopictus} is possible in captivity and in the wild \cite{arunachalam2008natural,bosio1992variation,cecilio2009natural,gunther2007evidence,rosen1983transovarial}. Thus, assessing transmission dynamics and pathogenicity between the DENV-2 American and Asian genotypes' differences is one of the priorities associated with the study of the epidemiology of dengue. In short, dengue has an increasing  recurrent  presence  putting a larger percentage of the  global population at  risk of dengue infection, a situation that has become the norm due to the growth of travel and tourism between endemic and non-endemic  regions. The aim of this work is to better understand the impact of human mobility on dengue disease transmission, its impact on dengue dynamics, and the use mobility based strategies, an standard control measures, in reducing the prevalence of dengue infections . \\

Mathematical models describing the dynamics of interaction between host and vector go back to Ross \cite{Ross1911}, Lotka \cite{Lotka1923} and MacDonald \cite{macdonald1952analysis};  first used to study vector-host dynamics in the context of Malaria \cite{brauer2011mathematical,gumel2006mathematical,shim2012differential}. Variations of such framework have been applied to dengue ( for a review see \cite{smith2012ross}). Further applications of modeling variations in the context of Malaria include,\cite{forouzannia2014mathematical,mckenzie2004role,ngwa2010mathematical,niger2008mathematical} and in the context of dengue \cite{cccsanchez2011,chowell2006climate,gumel2006mathematical,murillo2014vertical,nishiura2006mathematical}. \\

The potential role of vertical transmission in dengue endemic regions or in fluctuating environments has been explored in \cite{adams_how_2010,esteva2000influence,nishiura2006mathematical}.   The role in the displacement of DENV-2 American via DENV-2 Asian vertical transmission has also been addressed \cite{murillo2014vertical}. The role of host movement has also been explored in the context of dengue \cite{adams2009man} in a formulation that does not account for the the {\it effective population size}. In this paper, the role of vertical transmission and movement via residence times are explored via a two-patch model involving non-mobile vectors and mobile hosts. This paper is organized as follows: The derivation of the model is presented in Section \ref{sec:modelderiv}; Analytical results are collected in Section \ref{sec:analytical}; The results of numerical simulations are collected in Section \ref{sec:simula}; Section \ref{sec:casestudy} explores the possible role of movement on  joint dynamics of dengue in Colima and Manzanillo in the presence of host mobility;  Concluding remarks are collected in Section \ref{sec:conclusion}.

\section{Derivation of the model}
\label{sec:modelderiv}

A single patch model is derived and embedded into a two-patch model via a residence-times matrix in order to study the impact of host mobility on dengue disease dynamics. Conditions for dengue eradication and persistence in the population are computed.
\subsection{Single patch model}
\label{sec:singlepatch}
We consider a population of host composed of susceptible ($S_h$), exposed ($E_h$), infectious ($I_h$) and recovered $(R_h)$ individuals interacting with a vector population composed of susceptible ($S_v$), exposed ($E_v$) and infected ($I_v$) vectors. The dynamics of dengue follows an $SEIR$ structure for the host population and an $SEI$ type for the vector population.  The birth rate for the host population is $\mu_h$, assumed to be equal to the death rate, that is, hosts' demographic differentials are conveniently ignored, that is, the host population is assumed to be constant. Susceptible hosts are infected, by infectious mosquitoes, at the rate $a \beta_{vh}\frac{I_{v}}{N_{h}}$ where $a$ is the biting rate and $\beta_{vh}$  is the infectiousness of human to mosquitoes. The exposed population develops symptoms  becoming infectious at the rate $\nu_h$. Infectious individuals recover at the per-capita rate $\gamma$. Susceptible mosquitoes become infected, via interactions with infectious hosts, at the rate $a\beta_{hv}\frac{ I_{h}}{ N_{h}}$.  Recent studies place significant importance to the connection between DENV-2 and DHF cases \cite{chowell2007clinical,espinoza2005clinical,montoya-circulation,rico1997origins,sittisombut1997possible,zhang2006structure} and on DENV-2 vertical transmission \cite{martins2012occurrence}. Hence, it is assumed that a fraction $q$  of the mosquitoes are ``born" infected  entering directly the infectious class. The natural per-capita vector mortality is $\mu_v$. \\

The model describing the dynamics of DENV-2 is given by the following system of differential equations:
\begin{equation} \label{1Patch}
\left\{\begin{array}{llll}
\dot S_{h}=\mu_hN_{h}-a \beta_{vh}S_{h}\frac{I_{v}}{N_{h}}
-\mu_h S_{h}\\
\dot E_{h}=\beta_{vh}S_h\frac{I_{v}}{N_{h}}-(\mu_h+\nu_h) E_{h}\\
\dot I_{h}=\nu_hE_{h}-(\mu_h+\gamma_h) I_{h}\\
\dot R_{h}=\gamma_hI_h-\mu_h R_h\\
\dot S_{v}=\mu_v(N_v-qI_v)-a\beta_{hv}S_{v}\frac{ I_{h}}{ N_{h}}-\mu_vS_{v}\\
\dot{E}_{v}=a\beta_{hv}S_{v}\frac{ I_{h}}{ N_{h}}-(\nu_v+\mu_v)E_{v}\\
\dot{I}_{v}=\nu_v E_v+q\mu_vI_v-\mu_v I_v
\end{array}\right.
\end{equation}
In the absence of selection, that is, differences in birth and death rate and in the absence of vertical transmission, Model (\ref{1Patch}) turns out to be isomorphic to model considered by Chowell et \textit{al} in \cite{chowell_estimation_2007}. Model (\ref{1Patch}) is well defined supporting a sharp threshold property, namely, the disease dies out if the basic reproduction number $\mathcal R_0$ is less than unity, persisting whenever $\mathcal R_{0}>1$ where 

$$\mathcal R_0^2=\frac{a^2\beta_{hv}\beta_{vh}N_{v}\nu_h\nu_v}{(1-q)N_h(\mu_h+\nu_h)(\mu_h+\gamma_h)(\mu_v+\nu_v)\mu_v}.$$

\subsection{Heterogeneity through virtual dispersal}
\label{sec:twopatch}
The single patch model is the building  block for the two-patch model used in this study. Within each patch, in the absence of host mobility, dengue dynamics are modeled  via System \ref{1Patch}.   A metapopulation approach, an Eulerian perspective, is most often applied to the study of vector-borne diseases involving host mobility (\cite{adams2009man, auger2008ross,gao2012multipatch}). Here, a Lagrangian approach is used instead to model the movement of individuals between patches (see \cite{bichara2015vector,bichara2015sis}). It is assumed that vectors don't move between patches since vecors \textit{Ae. aegypti} and \textit{Ae. albopictus} do not travel more than few tens of meters over their lifetime \cite{adams2009man,WHODengueMosquitoes}; moving 400-600 meters at most \cite{bonnet1946dispersal,niebylski1994dispersal}, respectively. In short, we neglect vector's dispersal, which fits well the simulations involving two cities in the state of Colima, Mexico.\\

The host resident of Patch 1, population size $N_{h,1}$, spends, on average, $p_{11}$ proportion of its time in their own Patch 1 and $p_{12}$ proportion of its time visiting Patch 2. Residents of Patch 2, population of size $N_{h,2}$, spend $p_{22}$ proportion of their time in Patch 2 while spending $p_{21}=1-p_{22}$ visiting Patch 1. Thus, at time $t$, the \textit{effective population} in Patch 1 is $p_{11}N_{h,1}+p_{21}N_{h,2}$ and the \textit{effective population} in Patch 2 is $p_{12}N_{h,1}+p_{22}N_{h,2}$. The susceptible population of Patch 1 ($S_1$) could be infected by a vector in either Patch 1 ($I_{v,1}$) or  Patch 2 by ($I_{v,2}$). Thus, the dynamics of the susceptible population in Patch 1 are given by
\begin{equation}
\dot S_{h,1}=\mu_{h}N_{h,1}-a_1\beta_{vh}p_{11}S_{h,1}\frac{I_{v,1}}{p_{11}N_{h,1}+p_{21}N_{h,2}}-a_2\beta_{vh}p_{12}S_{h,1}\frac{I_{v,2}}{p_{12}N_{h,1}+p_{22}N_{h,2}}-\mu_{h} S_{h,1}.\end{equation} 
\\
And so, the  \textit{effective infectious} population in Patch 1 is $p_{11}I_{h,1}+p_{21}I_{h,2}$ and, consequently, the proportion of infectious individuals in Patch 1, is
$$\frac{p_{11}I_{h,1}+p_{21}I_{h,2}}{p_{11}N_{h,1}+p_{21}N_{h,2}}.$$

The dynamics of susceptible mosquitoes in Patch 1 are modeled as follows:

\begin{equation}
\dot{S}_{v,1}=\mu_{v}(N_{v,i}-qI_{v,i})-a_1\beta_{hv}S_{v,1}\frac{p_{11}I_{h,1}+p_{21}I_{h,2}}{p_{11}N_{h,1}+p_{21}N_{h,2}}-\mu_vS_{v,1}.\end{equation}
\\
The complete dynamics of DENV-2, with the host moving between patches, is given by the following system,
\begin{equation} \label{PatchGenConsR}
\left\{\begin{array}{llll}
\dot S_{h,i}=\mu_hN_{h,i}-\beta_{vh}S_{h,i}\sum_{j=1}^{2} a_jp_{ij}\frac{I_{v,j}}{p_{1j}N_{h,1}+p_{2j}N_{h,2}}-\mu_h S_{h,i},\\
\dot E_{h,i}=\beta_{vh}S_{h,i}\sum_{j=1}^{2} a_jp_{ij}\frac{I_{v,j}}{p_{1j}N_{h,1}+p_{2j}N_{h,2}}-(\mu_h+\nu_h) E_{h,i},\\
\dot I_{h,i}=\nu_h E_{h,i}-(\mu_h+\gamma_i) I_{h,i},\\
\dot R_{h,i}=\gamma_i I_{h,i}-\mu_hR_{h,i},\\
\dot S_{v,i}=\mu_{v}(N_{v,i}-qI_{v,i})-a_i\beta_{hv}S_{v,i}\frac{ \sum_{j=1}^{2}p_{ji}I_{h,j}}{ \sum_{k=1}^{2}p_{ki}N_{h,k}}-\mu_vS_{v,i},\\
\dot{E}_{v,i}=a_i\beta_{hv}S_{v,i}\frac{ \sum_{j=1}^{2}p_{ji}I_{h,j}}{ \sum_{k=1}^{2}p_{ki}N_{h,k}}-(\mu_v+\nu_v)E_{v,i},\\
\dot{I}_{v,i}=\nu_vE_{h,i}+q\mu_vI_{v,i}-\mu_vI_{v,i}, i=1,2.
\end{array}\right.
\end{equation}
\\
Since the total populations of  hosts and vectors are constant in each patch,
 System (\ref{PatchGenConsR}) has the same qualitative dynamics as,

\begin{equation} \label{PatchGenFinal}
\left\{\begin{array}{llll}
\dot S_{h,i}=\mu_hN_{h,i}-\beta_{vh}S_{h,i}\sum_{j=1}^{2} a_jp_{ij}\frac{I_{v,j}}{p_{1j}N_{h,1}+p_{2j}N_{h,2}}-\mu_h S_{h,i},\\
\dot E_{h,i}=\beta_{vh}S_{h,i}\sum_{j=1}^{2} a_jp_{ij}\frac{I_{v,j}}{p_{1j}N_{h,1}+p_{2j}N_{h,2}}-(\mu_h+\nu_h) E_{h,i},\\
\dot I_{h,i}=\nu_h E_{h,i}-(\mu_h+\gamma_i) I_{h,i},\\
\dot{E}_{v,i}=a_i\beta_{hv}(N_{v,i}-E_{v,i}-I_{v,i})\frac{ \sum_{j=1}^{2}p_{ji}I_{h,j}}{ \sum_{k=1}^{2}p_{ki}N_{h,k}}-(\mu_v+\nu_v)E_{v,i},\\
\dot{I}_{v,i}=\nu_vE_{h,i}-(1-q)\mu_vI_{v,i}.
\end{array}\right.
\end{equation}
\begin{table}[h!]
  \begin{center}
    \caption{Description of the parameters used in System (\ref{PatchGenFinal}).}
    \label{tab:Param}
    \begin{tabular}{cc}
      \toprule
      Parameters & Description \\
      \midrule
$\beta_{vh}$ & Infectiousness of human to mosquitoes\\
$\beta_{hv}$ & Infectiousness of mosquitoes to humans\\
$a_{i}$ &  Biting rate in Patch $i$ \\
$\mu_h$ & Humans' birth and death rate\\
$\nu_h$ & Humans' incubation rate\\
$\gamma_i$  & Recovery rate in Patch $i$\\
$p_{ij}$ &  Proportion of time residents of Patch $i$ spend in Patch $j$\\
$\mu_v$  & Vectors' natural birth and mortality rate\\
$\nu_v$ & Vectors' incubation rate\\
$q$  & Proportion of mosquitoes infected through transovarial transmission \\
      \bottomrule
    \end{tabular}
  \end{center}
\end{table}
The parameters of Model \ref{PatchGenFinal} are described in Table \ref{tab:Param}.\\

We now show that the model is {\it biologically} well posed.

\begin{lemma}\label{Boundedness} \hfill \\
The set $$\Omega=\{(S_{h,i},E_{h,i},I_{h,i}, E_{v,i},I_{v,i})\in\mathbb R^{6}_+ \quad\mid\quad S_{h,i}+E_{h,i}+I_{h,i}\leq N_{h,i}, \; E_{v,i}+I_{v,i}\leq N_{v,i} \}$$
 is a compact positively invariant for the System (\ref{PatchGenFinal}).
\end{lemma}

\begin{proof}\hfill

The positive orthant is clearly positively invariant. Since the host population is constant, then the inequality $S_i+E_i+I_i\leq N_{h,i}$ is always satisfied. We have 
\begin{eqnarray}
\dot E_{v,i}+\dot I_{v,i}\mid_{E_{h,i}=N_{v,i}}&=&-a_i\beta_{hv}I_{v,i}\frac{ \sum_{j=1}^{2}p_{ji}I_{h,j}}{ \sum_{k=1}^{2}p_{ki}N_{h,k}}-(\mu_v+\nu_v)N_{v,i}\nonumber\\
&\leq& 0\nonumber
\end{eqnarray}
Hence, $E_{v,i}+I_{v,i}\leq N_{v,i}$ and the set $\Omega$, an intersection of positively invariant sets ( $\mathbb R^7_+$, $\{S_i+E_i+I_i\leq N_{h,i}\}$ and $\{E_{v,i}+I_{v,i}\leq N_{v,i}\}$), is positively invariant; the set is a compact set.\\
\end{proof}

\section{Equilibria and stability analysis}
\label{sec:analytical}
This section characterizes the equilibrium dynamics of Model (\ref{PatchGenFinal}).

\subsection{The disease free equilibrium and the basic reproduction number}
The disease free equilibrium is $$E_0=(N_{h,1},N_{h,2},\textbf{0}_{\mathbb R^{8}}),$$
which is used to compute the basic reproduction number via the next generation method \cite{MR1057044,VddWat02}. 
The basic reproduction number $\mathcal R_0$ is defined by the expression (See Appendix \ref{R0}, for details), $\mathcal R_0^2=\rho(M_{vh}M_{hv})$, that is, the spectral radius of the matrix of $M_{vh}M_{hv}$,
where 
$$
   M_{vh}= \left(    \begin{array}{cc}\frac{a_1\beta_{vh}p_{11} N_{h,1}\nu_v}{(p_{11}N_{h,1}+p_{21}N_{h,2})(\mu_v+\nu_v)(1-q)\mu_v} & \frac{a_2\beta_{vh}p_{12}N_{h,1}\nu_v}{(p_{12}N_{h,1}+p_{22}N_{h,2}) (\mu_v+\nu_v)(1-q)\mu_v}\\ \frac{a_1\beta_{vh}p_{21}N_{h,2}\nu_v}{ (p_{11}N_{h,1}+p_{21}N_{h,2})(\mu_v+\nu_v)(1-q)\mu_v} & \frac{a_2\beta_{vh}p_{22}N_{h,2}\nu_v}{ (p_{12}N_{h,1}+p_{22}N_{h,2})(\mu_v+\nu_v)(1-q)\mu_v}\end{array}\right)
$$
and 
$$
   M_{hv}= \left(    \begin{array}{cc}\frac{a_1\beta_{hv}p_{11}  N_{v,1}\nu_h}{(p_{11}N_{h,1}+p_{21}N_{h,2})(\mu_h+\nu_h)(\mu_h+\gamma_1)} & \frac{a_1\beta_{hv}p_{21} N_{v,1}\nu_h}{(p_{11}N_{h,1}+p_{21}N_{h,2}) (\mu_h+\nu_h)(\mu_h+\gamma_2)}\\
   \frac{a_2\beta_{hv}p_{12} N_{v,2}\nu_h}{ (p_{12}N_{h,1}+p_{22}N_{h,2})(\mu_h+\nu_h)(\mu_h+\gamma_1)} & \frac{a_2\beta_{hv}p_{22} N_{v,2}\nu_h}{ (p_{12}N_{h,1}+p_{22}N_{h,2})(\mu_h+\nu_h)(\mu_h+\gamma_2)}
   \end{array}\right).
$$

The matrix $\left(    \begin{array}{cc}0 & M_{vh}\\
   M_{hv} & 0
   \end{array}\right)$ is called the host-vector network configuration \cite{iggidr2014dynamics}. The result of local asymptotic stability if $\mathcal R_0^2<1$ and instability if $\mathcal R_0^2>1$ has been established in \cite{VddWat02}. The following theorem gives the global result of the DFE. 
\begin{theorem}
If $\mathcal R_0^2\leq1$, the DFE is globally asymptotically stable in the nonnegative
orthant. If $\mathcal R_0^2>1$, the DFE is unstable.
\end{theorem}
\begin{proof}\hfill 

We use the comparison theorem \cite{cite-SmitWalt95} to prove the GAS of the DFE. Since $S_{h,i}\leq N_{h,i}$ and $S_{v,i}\leq N_{v,i}$, we have that,
\begin{equation}\label{EhGAS}\dot E_{h,i}\leq\beta_{vh}N_{h,i}\sum_{j=1}^{2} a_jp_{ij}\frac{I_{v,j}}{p_{1j}N_{h,1}+p_{2j}N_{h,2}}-(\mu_h+\nu_h) E_{h,i}\end{equation} 
and 
\begin{equation}\label{EvGAS}\dot{E}_{v,i}\leq a_i\beta_{hv}N_{v,i}\frac{ \sum_{j=1}^{2}p_{ji}I_{h,j}}{ \sum_{k=1}^{2}p_{ki}N_{h,k}}-(\mu_v+\nu_v)E_{v,i}.\end{equation}
We define an auxiliary system via the right hand side of Equations  (\ref{EhGAS})-(\ref{EvGAS})  and the infected compartments of Equation (\ref{PatchGenFinal}) as follows:
 \begin{eqnarray}\label{Comp}\left(\begin{array}{c}
\dot E_{h,i}\\
\dot E_{v,i}\\
\dot I_{h,i}\\
\dot I_{v,i}
\end{array}\right)&=&\left(\begin{array}{c}
\beta_{vh}N_{h,i}\sum_{j=1}^{2} a_jp_{ij}\frac{I_{v,j}}{p_{1j}N_{h,1}+p_{2j}N_{h,2}}-(\mu_h+\nu_h) E_{h,i}\\
 a_i\beta_{hv}N_{v,i}\frac{ \sum_{j=1}^{2}p_{ji}I_{h,j}}{ \sum_{k=1}^{2}p_{ki}N_{h,k}}-(\mu_v+\nu_v)E_{v,i}\\
 \nu_h E_{h,i}-(\mu_h+\gamma_i) I_{h,i}\\
\nu_vE_{h,i}-(1-q)\mu_vI_{v,i}
\end{array}\right)\nonumber\\
&=&
 (F+V)\left(\begin{array}{c}
 E_{h,i}\\
E_{v,i}\\
 I_{h,i}\\
 I_{v,i}
\end{array}\right);
\end{eqnarray}
where the matrices $F$ and $V$ in (\ref{Comp}) were just generated using the next generation method. System (\ref{Comp}) is linear and its dynamics is well known. Since $V$ is a Metzler matrix and $F$ a nonnegative matrix (\cite{berman1979nonnegative}), then 
$$\rho(-FV^{-1})<1\iff \alpha(F+V)<0$$
where $ \alpha(F+V)$ is the stability modulus of $F+V$. Thus, if $\mathcal R_0=\rho(-FV^{-1})<1$, all the eigenvalues of $F+V$ are negative. Hence, the nonnegative solutions of (\ref{Comp}) are such that $$\lim_{t\to\infty}E_{h,i}=\lim_{t\to\infty}E_{v,i}=0\quad\text{and}\quad\lim_{t\to\infty}I_{h,i}=\lim_{t\to\infty}I_{v,i}=0$$
Since, all the variables in System (\ref{PatchGenFinal}) are nonnegative, the use of a comparison theorem \cite{cite-SmitWalt95} leads to,
$$\lim_{t\to\infty}E_{h,i}=\lim_{t\to\infty}E_{v,i}=0\quad\text{and}\quad\lim_{t\to\infty}I_{h,i}=\lim_{t\to\infty}I_{v,i}=0, \quad  i=1,2.$$
Therefore, by using the asymptotic theory of autonomous systems \cite{CasThi95}, System (\ref{PatchGenFinal}) has the qualitative dynamics of the following limit system:
$$
\dot S_{h,i}=\mu_hN_{h,i}-\mu_h S_{h,i}
$$
for which the equilibrium $(N_{h,1},N_{h,2})$ is globally asymptotically stable. If $\mathcal R_0>1$, the instability of the DFE follows from  \cite{MR1057044,VddWat02}.
\end{proof}
\begin{theorem}
If $\mathcal R_0>1$, System (\ref{PatchGenFinal}) is uniformly persistent, that is, it exists $\eta>0$ such that $\displaystyle \liminf_{t\to\infty}\{S_{h,i},E_{h,i},I_{h,i}, E_{v,i}, I_{v,i}\}>\eta$ for any initial conditions satisfying $S_{h,i}(0)>0$, $E_{h,i}(0)>0,$ $I_{h,i}(0)>0$, $E_{v,i}(0)>0$ and $I_{v,i}(0)>0$ for $i=1, 2$.
\end{theorem}
\begin{proof}
Let $X=\Omega$, $x= (S_{h,1}, S_{h,2}, E_{h,1}, E_{h,2}, E_{v,1}, E_{v,2}, I_{h,1}, I_{h,2}, I_{v,1}, I_{v,2})$ and $X_0=\{  x\in X \;\mid\; I_{v,1}+I_{v,2}>0  \}$
Hence, $\partial X_0=X\backslash X_0=\{x\in X \;\mid\; I_{v,1}=I_{v,2}=0\}$. Let $\phi_t$ be semi-flow induced by the solutions of (\ref{PatchGenFinal}) and $M_\partial=\{ x\in\partial X_0\; \mid\; \phi_tx\in\partial X_0, t\geq0 \}$. By Lemma \ref{Boundedness}, we have $\phi_tX_0\subset X_0$ and $\phi_t$ is bounded in $X_0$. Therefore a global attractor for $\phi_t$ exists . The DFE is the unique equilibrium on the manifold $\partial X_0$ and is GAS on $\partial X_0$. Moreover $\cup_{x\in M_\partial} \omega (x)=\{E_0\}$ and no subset of $M$ forms a cycle in $\partial X_0$. Finally since the DFE is unstable on $X_0$ if $\mathcal R_0>1$, we deduce that System (\ref{PatchGenFinal}) is uniformly persistent by using a result from \cite{zhao2013dynamical} (Theorem 1.3.1 and Remark 1.3.1).
\end{proof}
\begin{theorem}
Whenever the host-vector configuration is irreducible and $\mathcal R_0^2>1$, System (\ref{PatchGenFinal}) has a unique endemic equilibrium.
\end{theorem}
\begin{proof}
We will use a result by Hethcote and Thieme \cite{hethcote1985stability} to prove the uniqueness of the endemic equilibrium. An endemic equilibrium $(\bar S_{h,1},\bar S_{h,2},\bar E_{h,1},\bar E_{h,2},\bar E_{v,1},\bar E_{v,2},\bar I_{h,1},\bar I_{h,2},\bar I_{v,1},\bar I_{v,2})$ satisfies:
\begin{equation} \label{EEl}
\left\{\begin{array}{llll}
\mu_hN_{h,i}=\beta_{vh}\bar S_{h,i}\sum_{j=1}^{2} a_jp_{ij}\frac{\bar I_{v,j}}{p_{1j}N_{h,1}+p_{2j}N_{h,2}}+\mu_h \bar S_{h,i},\\
(\mu_h+\nu_h) \bar E_{h,i}=\beta_{vh}\bar S_{h,i}\sum_{j=1}^{2} a_jp_{ij}\frac{\bar I_{v,j}}{p_{1j}N_{h,1}+p_{2j}N_{h,2}},\\
\nu_h \bar E_{h,i}=(\mu_h+\gamma_i)\bar I_{h,i},\\
(\mu_v+\nu_v)\bar E_{v,i}=a_i\beta_{hv}( N_{v,i}-\bar E_{v,i}-\bar I_{v,i})\frac{ \sum_{j=1}^{2}p_{ji}\bar I_{h,j}}{ \sum_{k=1}^{2}p_{ki}N_{h,k}},\\
(1-q)\mu_vI_{v,i}=\nu_vE_{v,i}.
\end{array}\right.
\end{equation}
The first equation of (\ref{EEl}) implies that $$\bar S_{h,i}=\frac{\mu_h N_{h,i}}{\beta_{vh}\sum_{j=1}^{2} a_jp_{ij}\frac{\bar I_{v,j}}{p_{1j}N_{h,1}+p_{2j}N_{h,2}}+\mu_h}.$$
Hence, we deduce that, from System (\ref{EEl}), that
\begin{equation} \label{EEl2}
\left\{\begin{array}{llll}
 \bar E_{h,i}=\frac{\beta_{vh}}{\mu_h+\nu_h}\frac{\mu_h N_{h,i}}{\beta_{vh}\sum_{j=1}^{2} a_jp_{ij}\frac{\bar I_{v,j}}{p_{1j}N_{h,1}+p_{2j}N_{h,2}}+\mu_h }\sum_{j=1}^{2} a_jp_{ij}\frac{\bar I_{v,j}}{p_{1j}N_{h,1}+p_{2j}N_{h,2}},\\
 \bar I_{h,i}=\frac{\nu_h}{\mu_h+\gamma_i}\bar E_{h,i},\\
\bar E_{v,i}=\frac{a_i\beta_{hv}}{\mu_v+\nu_v}(N_{v,i}-\bar E_{v,i}-\bar I_{v,i})\frac{ \sum_{j=1}^{2}p_{ji}\bar I_{h,j}}{ \sum_{k=1}^{2}p_{ki}N_{h,k}},\\
\bar I_{v,i}=\frac{\nu_v}{(1-q)\mu_v}\bar E_{v,i}.
\end{array}\right.
\end{equation}

Let 
$$F(x)=\left(\begin{array}{c}
\frac{\beta_{vh}\nu_v}{(1-q)(\mu_h+\nu_h)\mu_v}\frac{\mu_h N_{h,1}}{\beta_{vh}\sum_{j=1}^{2} \frac{a_jp_{1j}\nu_v}{(1-q)\mu_v}\frac{\bar E_{v,j}}{p_{1j}N_{h,1}+p_{2j}N_{h,2}}+\mu_h}\sum_{j=1}^{2} a_jp_{1j}\frac{\bar E_{v,j}}{p_{1j}N_{h,1}+p_{2j}N_{h,2}}\\

\frac{\beta_{vh}\nu_v}{(1-q)(\mu_h+\nu_h)\mu_v}\frac{\mu_h N_{h,2}}{\beta_{vh}\sum_{j=1}^{2} \frac{a_jp_{2j}\nu_v}{(1-q)\mu_v}\frac{\bar E_{v,j}}{p_{1j}N_{h,1}+p_{2j}N_{h,2}}+\mu_h}\sum_{j=1}^{2} a_jp_{2j}\frac{\bar E_{v,j}}{p_{1j}N_{h,1}+p_{2j}N_{h,2}}\\

\frac{a_1\beta_{hv}\nu_h}{(\mu_v+\nu_v)(\mu_h+\gamma_1)}(N_{v,1}-\bar E_{v,1}-\frac{\nu_v}{(1-q)\mu_v}\bar E_{v,1})\frac{ \sum_{j=1}^{2}p_{j1}\bar E_{h,j}}{ \sum_{k=1}^{2}p_{k1}N_{h,k}}\\

\frac{a_2\beta_{hv}\nu_h}{(\mu_v+\nu_v)(\mu_h+\gamma_2)}(N_{v,2}-\bar E_{v,2}-\frac{\nu_v}{(1-q)\mu_v}\bar E_{v,2})\frac{ \sum_{j=1}^{2}p_{j2}\bar E_{h,j}}{ \sum_{k=1}^{2}p_{k2}N_{h,k}}

\end{array}\right)$$

where $x=(\bar E_{h,1},\bar E_{h,2},\bar E_{v,1},\bar E_{v,2},\bar I_{h,1},\bar I_{h,2})$. The function $F(x)$ is continuous, bounded, differentiable and $F(0_{\mathbb R^6})=0_{\mathbb R^6}$.  The function $F$ is monotone if the corresponding Jacobian matrix is Metzler, i.e all off-diagonal entries are nonnegative. We have:

$$DF(x)=\left(\begin{array}{cc}
 \begin{array}{cc}0 & 0\\ 0 & 0\end{array}  & \tilde M_{vh}(x)\\
 \tilde M_{hv}(x) &   \begin{array}{cc}-a_1\beta_{hv}\left(1+\frac{\nu_v}{(1-p)\mu_v}\right)\frac{  \sum_{k=1}^2\frac{ p_{k1}\nu_hE_{h,k}}{\mu_h+\gamma_k}}{p_{11}N_{h,1}+p_{21}N_{h,2}}& 0\\ 0 & -a_2\beta_{hv}\left(1+\frac{\nu_v}{(1-p)\mu_v}\right)\frac{  \sum_{k=1}^2\frac{ p_{k2}\nu_hE_{h,k}}{\mu_h+\gamma_k}}{p_{12}N_{h,1}+p_{22}N_{h,2}}\end{array} 

\end{array}\right)
$$ 
where 
\begin{eqnarray}\tilde m_{vh}^{ij}(x)&=&\frac{\beta_{vh}^2\mu_h\nu_va_jp_{ij}N_{h,j}}{(1-p)\mu_v(\mu_h+\nu_h)\sum_{k=1}^2 p_{kj}N_{h,k}} \frac{1}{\frac{\beta_{vh}}{(1-q)\mu_v} \sum_{k=1}^2\frac{ a_kp_{jk}E_{v,k}}{p_{1k}N_{h,1}+p_{2k}N_{h,2}}+\mu_h}   \left[\frac{    }{  } 1 - \right.\nonumber\\
& & \left.\frac{\frac{\beta_{vh}}{(1-q)\mu_v}\sum_{k=1}^2\frac{ a_kp_{jk}E_{v,k}}{p_{1k}N_{h,1}+p_{2k}N_{h,2}} }{\frac{\beta_{vh}}{(1-q)\mu_v} \sum_{k=1}^2\frac{ a_kp_{jk}E_{v,k}}{p_{1k}N_{h,1}+p_{2k}N_{h,2}}+\mu_h}\right]\end{eqnarray}

and
$$\tilde m_{hv}^{ij}(x)=a_i\beta_{hv}\left(N_{v,i}-E_{v,i}-\frac{\nu_v E_{v,i}}{(1-p)\mu_v}\right)\frac{\nu_hp_{ji}}{(\mu_h+\gamma_i)\sum_{k=1}^2p_{ki}N_{h,k}}.
$$
Since, $\tilde m_{vh}^{ij}\geq0$ and $\tilde m_{hv}^{ij}\geq0$ for all $i,j=1,2$, hence all off diagonal entries of the Jacobian matrix are nonnegative and so, the function $F(x)$ is monotone, moreover,

$$DF(0_{\mathbb R^4})=\left(\begin{array}{cc}
 \begin{array}{cc}0 & 0\\ 0 & 0\end{array}  & \tilde M_{vh}(0)\\
 \tilde M_{hv}(0) &  \begin{array}{cc}0 & 0\\ 0 & 0\end{array} 
\end{array}\right).
$$
This matrix is irreducible whenever $\tilde M_{vh}(0)\tilde M_{hv}(0)$ and $\tilde  M_{hv}(0)\tilde M_{vh}(0)$ are irreducible. The latter is guaranteed since $M_{vh}M_{hv}$ and $M_{hv}M_{vh}$ (from the next generation matrix) are both irreducible. Hence, an application of Theorem 2.1 in \cite{hethcote1985stability} implies that Model (\ref{EEl2}) has a unique positive fixed point \textit{if and only if} $\rho(DF(0_{\mathbb R^4}))=\mathcal R_0>1$, or equivalently $\mathcal R_0^2>1$.
\end{proof}

If the host-vector configuration is not irreducible, that is, the graphs associated with the matrices $M_{vh}M_{hv}$ and $M_{hv}M_{vh}$ are not strongly connected, the dynamics of the disease within patches are either somehow independent or System  (\ref{PatchGenFinal}) exhibits boundary equilibria. It is worthwhile noting that the irreducibility of residence times matrix $\mathbb P$ does not imply the irreducibility of $M_{vh}M_{hv}$ and $M_{hv}M_{vh}$. Since the epidemiological and entomological parameters are all positive, the reducibility of the host-vector configuration happens only on the three following cases: (i) If the two patches are isolated, i.e: $p_{12}=p_{21}=0$; (ii) residents of Patch 1 spend all their time in Patch 2 and residents of Patch 2 spend all their time in their own patch, i.e: $p_{12}=1$ and $p_{21}=0$; and (iii) the opposite scenario of (ii).

\section{Simulations}
\label{sec:simula}

Simulations are carried out in order to highlight the effects of residence times on disease dynamics. The simulations have a dual goal, first, to illustrate the theoretical results of this  manuscript and secondly to illustrate the impact of host mobility across high and low-risk dengue areas.  \\ 

The basic reproduction reproduction number $\mathcal R_0(\mathbb P)$ is a function of the residence times matrix $\mathbb P$. Simulation baseline values, except for those involving the entries of $\mathbb P$ are as follows:
$$
\beta_{hv} = 0.5(0.001-0.54),\; \beta_{vh} = 0.41(0.3-0.9),\; \frac{1}{\mu_v}=20(10-30)\; \textrm{days},\; a_1 = 0.95\; \textrm{day}^{-1}, a_2 = 0.8\; \textrm{day}^{-1},
$$
$$
\frac{1}{\mu_h}=60\times 365 \;\textrm{days},\; \frac{1}{\gamma_1}= 7\; \textrm{days},\; \frac{1}{\gamma_2}= 6\;  \textrm{days},\; \frac{1}{\nu_h}= 5\; \textrm{days}, \; \frac{1}{\nu_v}= 7\; \textrm{days}.
$$
The values of the parameters $\nu_h$ and $\nu_v$ are taken from \cite{adams_how_2010,adams2009man}. The infectiousness parameters ($\beta_{hv}$ and $\beta_{vh}$) and vector's natural mortality rate are taken from \cite{chitnis2013modelling}. Host and vector population are $$N_{h,1}=400,000\; N_{h,2}=300,000\; N_{v,1}=35,000,\; N_{v,2}=30,000$$
Patch 1 is the {\it high-risk} and Patch 2 is the {\it low-risk} and so, it is assumed that $a_1>a_2$. Figure \ref{fig:twofigsIh} represents the dynamics of Patch 1 (Fig \ref{Ih1Normal}) and Patch 2 (Fig \ref{Ih2Normal}) infected hosts while Fig \ref{fig:twofigsIv} collects the vector dynamics  in both patches. Since Patch 1 is  high-risk, the number of infected host should decrease as $p_{12}$ increases; see Fig \ref{Ih1Normal}. Figure \ref{Ih2Normal} shows the  Patch 2 infected host population, which it is decreasing, as $p_{21}$ and $p_{12} $ increase. Disease prevalence among Patch 2 residents remains very small when compared to that in Patch 1. In Fig \ref{fig:twofigsIv},  Patch 1 (Fig \ref{Iv1}) and Patch 2 (Fig \ref{Iv2}) vector dynamics are seen to follow the hosts' endemicity pattern.   \\

For all the different values of $p_{ij}$ chosen in Fig  \ref{fig:twofigsIh} and Fig  \ref{fig:twofigsIv}, the host-vector configuration matrix

$$M=\left(\begin{array}{cc}
0 & M_{vh}\\
M_{vh} & 0
\end{array}\right)$$
 or equivalently, the products $M_{hv}M_{vh}$ and $M_{vh}M_{hv}$, are irreducible. Moreover the basic reproduction number $\mathcal R_0$ is greater than one, hence the the disease is, in both patches, at an endemic level.

\begin{figure}[ht]
\centering
 \subfigure[The level of infected host in  Patch 1 seems to decrease as $p_{12}$ increases (and hence $p_{11}$ decreases).]{
   \includegraphics[scale =.43]{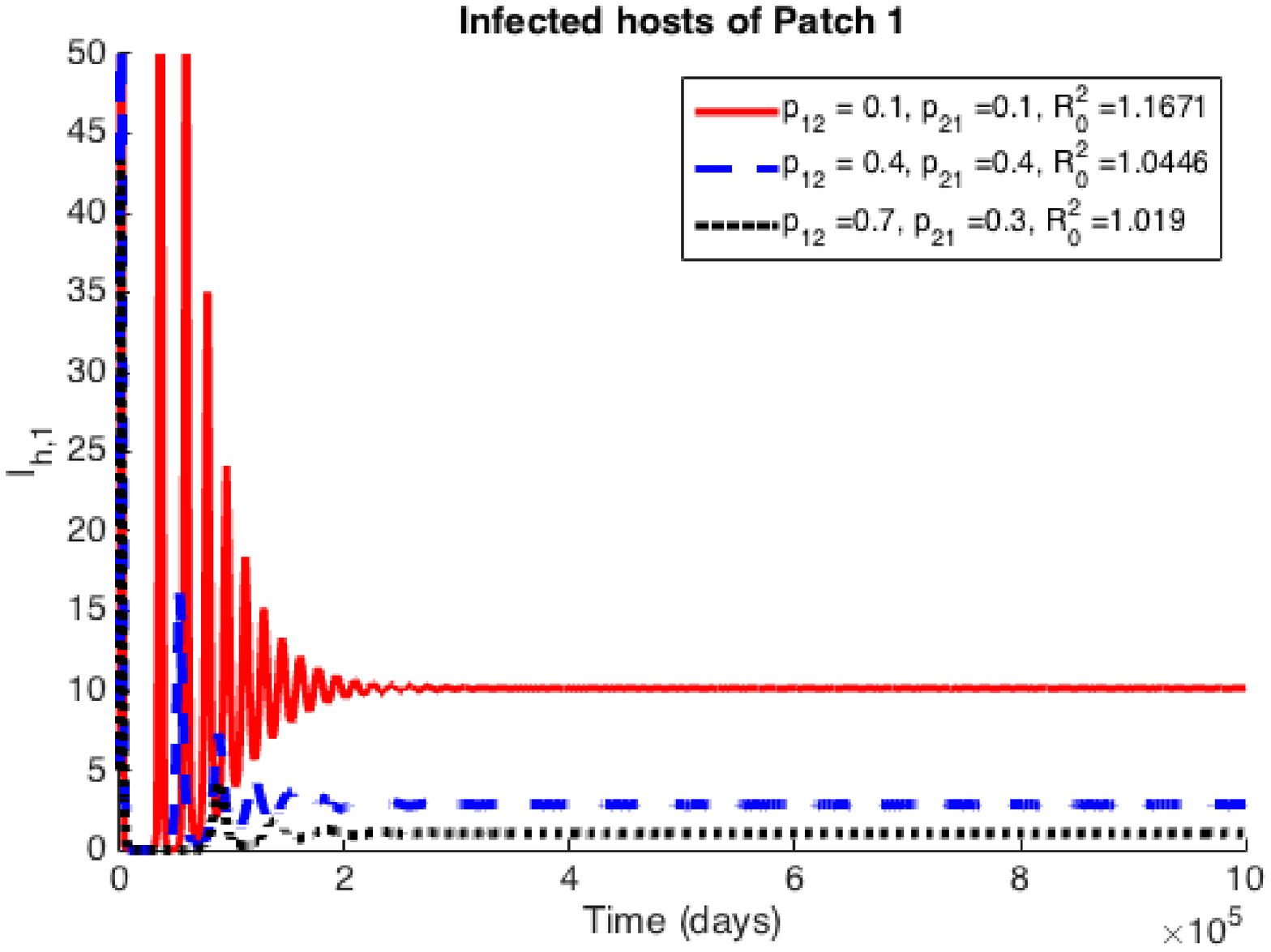}
\label{Ih1Normal}}
\hspace{1mm}
 \subfigure[The level of infected host in Patch 2 seems to decrease with respect to $p_{22}$.]{
   \includegraphics[scale =.43] {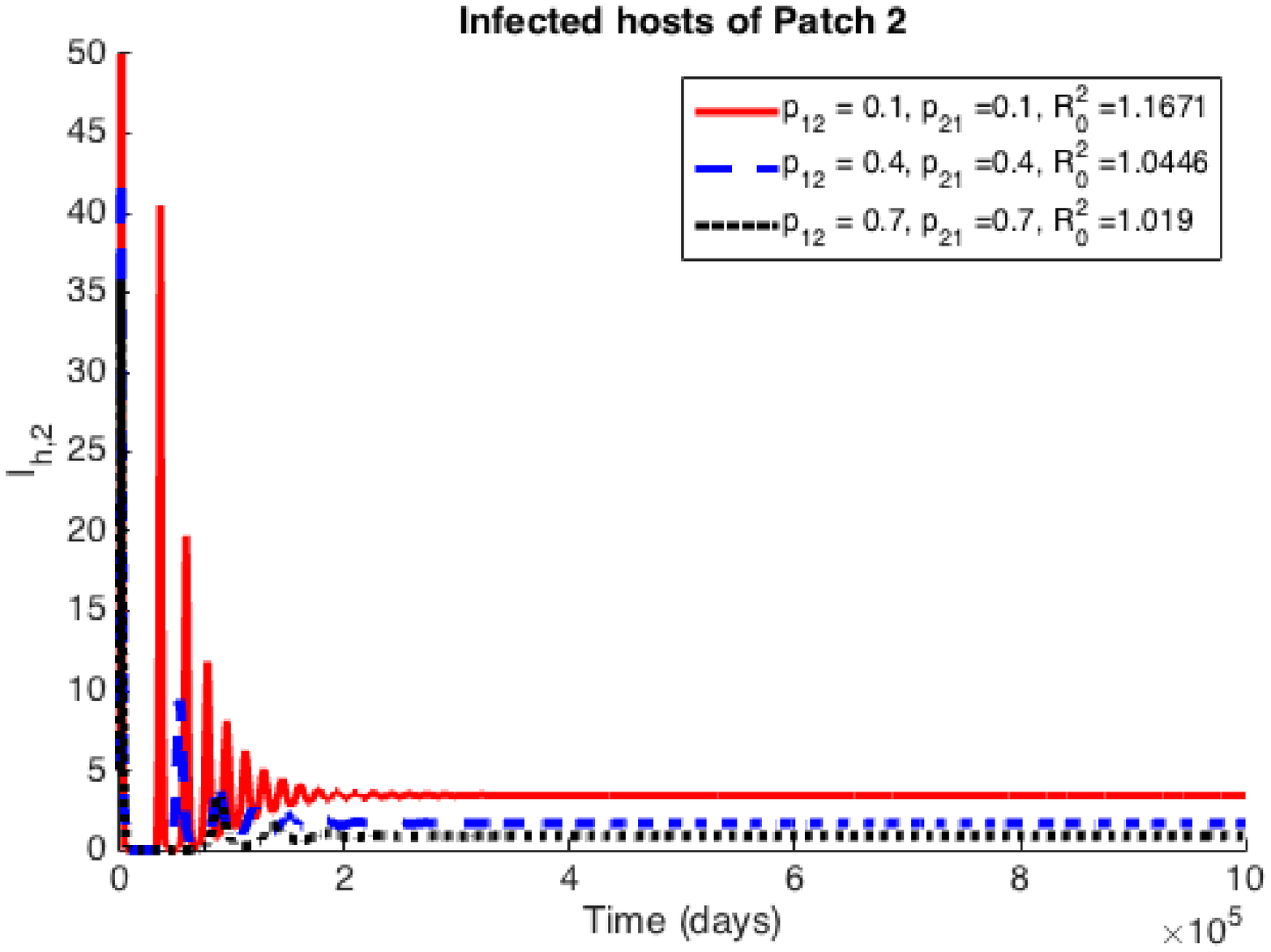}
\label{Ih2Normal}}
\caption{Dynamics of $I_{h,1}$ and $I_{h,2}$ for different values of $p_{ij}$.} \label{fig:twofigsIh}
\end{figure}

\begin{figure}[ht]
\centering
 \subfigure[Asymptotically, the level of infected vectors in  Patch 1 seems to decrease as $p_{12}$ increases (and hence $p_{11}$ decreases).]{
   \includegraphics[scale =.43]{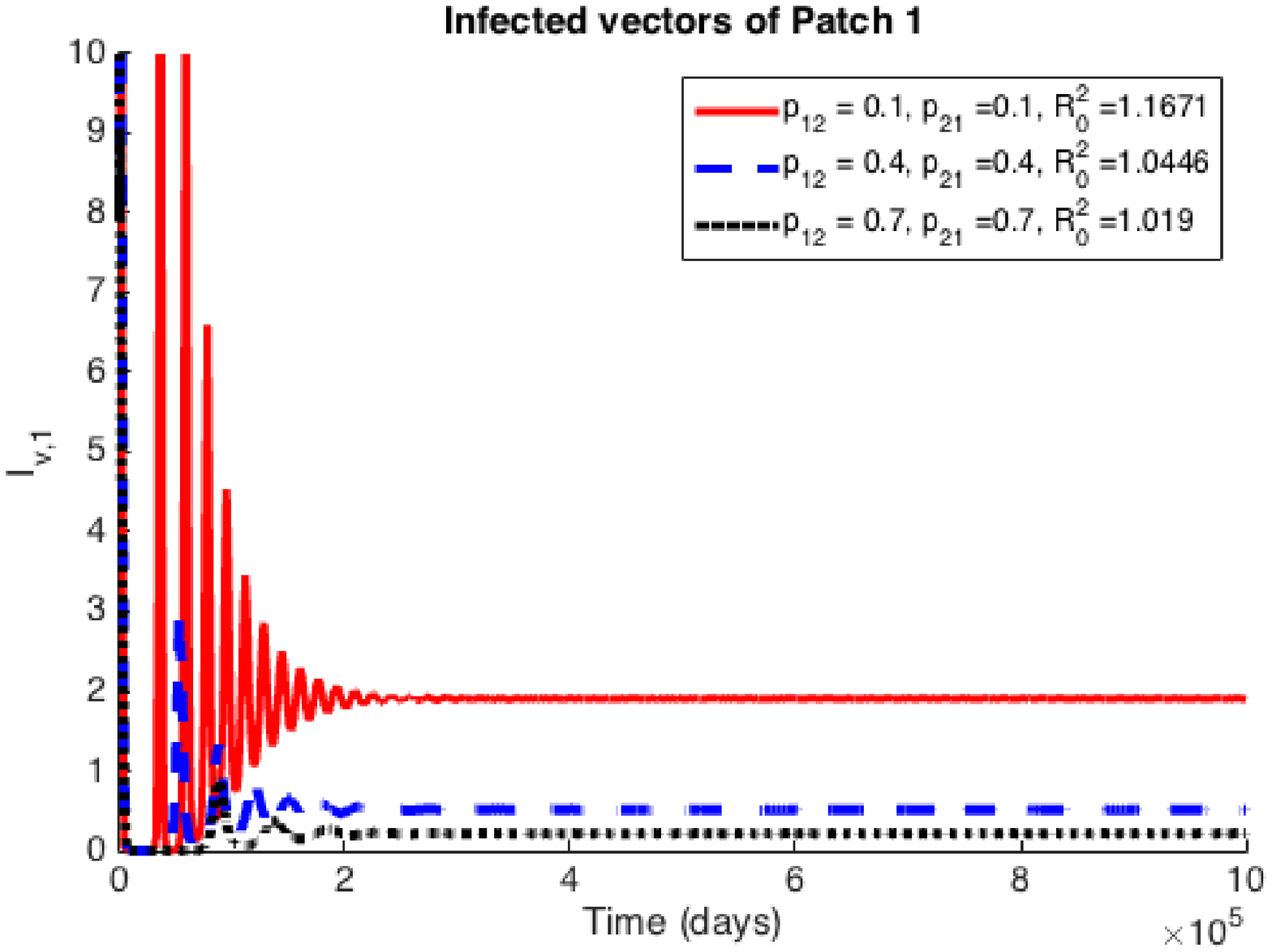}
\label{Iv1}}
\hspace{1mm}
 \subfigure[Asymptotically, the level of infected vectors in Patch 2 seems to decrease with respect to $p_{22}$.]{
   \includegraphics[scale =.43]{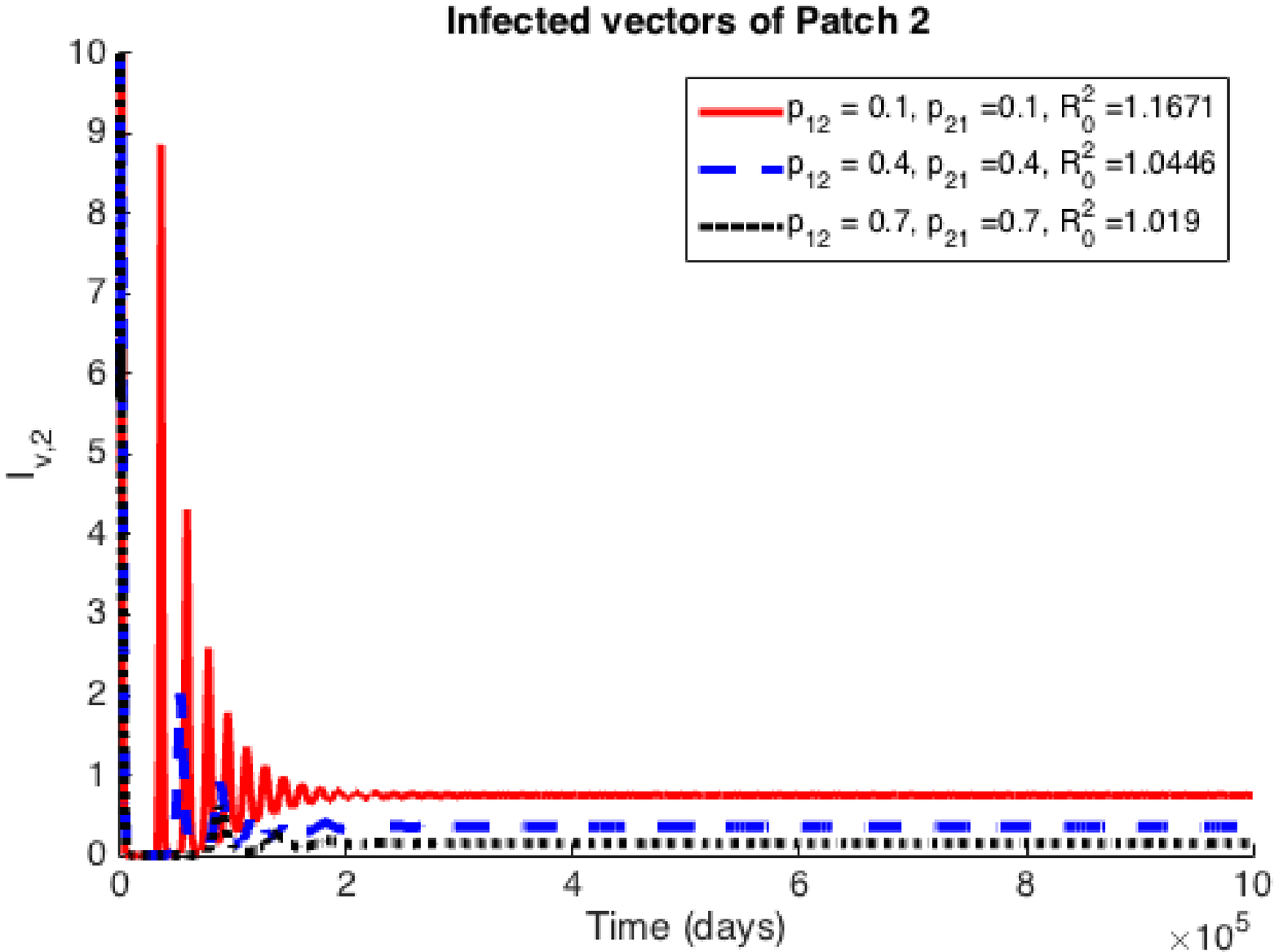}
\label{Iv2}}
\caption{Dynamics of $I_{v,1}$ and $I_{v,2}$ for different values of $p_{ij}$.} \label{fig:twofigsIv}
\end{figure}

Fig \ref{fig:twofigsIhReducible} displays the dynamics of the disease if the host-vector configuration matrix $M$ is not irreducible. The disease dies out in Patch 2 where the basic reproduction number is $\mathcal R_{2,0}^2=0.8161$ and persists in Patch 1 for which $\mathcal R_{1,0}^2=1.1747$.

\begin{figure}[ht]
\centering
 \subfigure{
  \includegraphics[scale =.43]{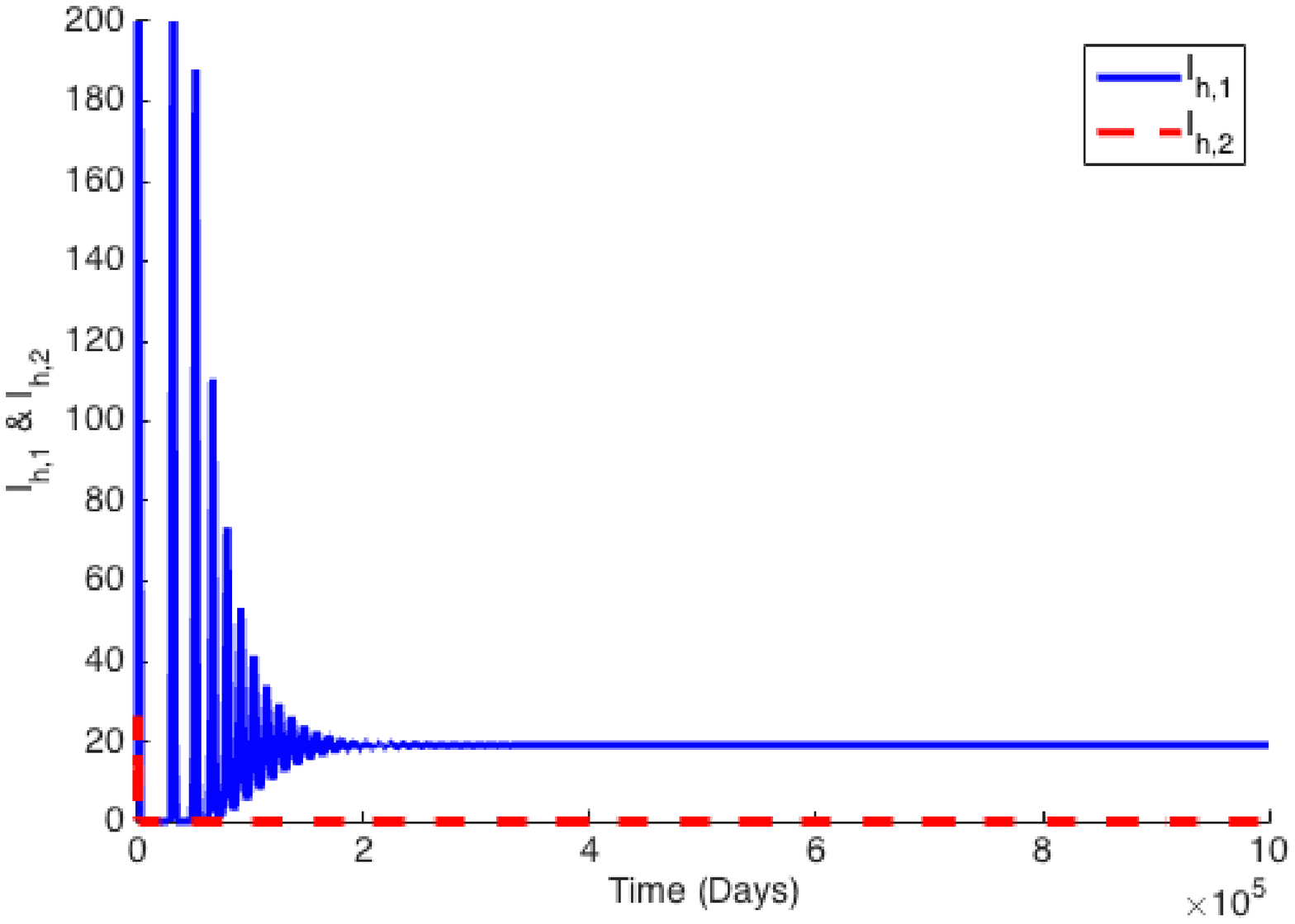}
\label{Ih1Ih2Reducible}}
\hspace{1mm}
 \subfigure{
   \includegraphics[scale =.43]{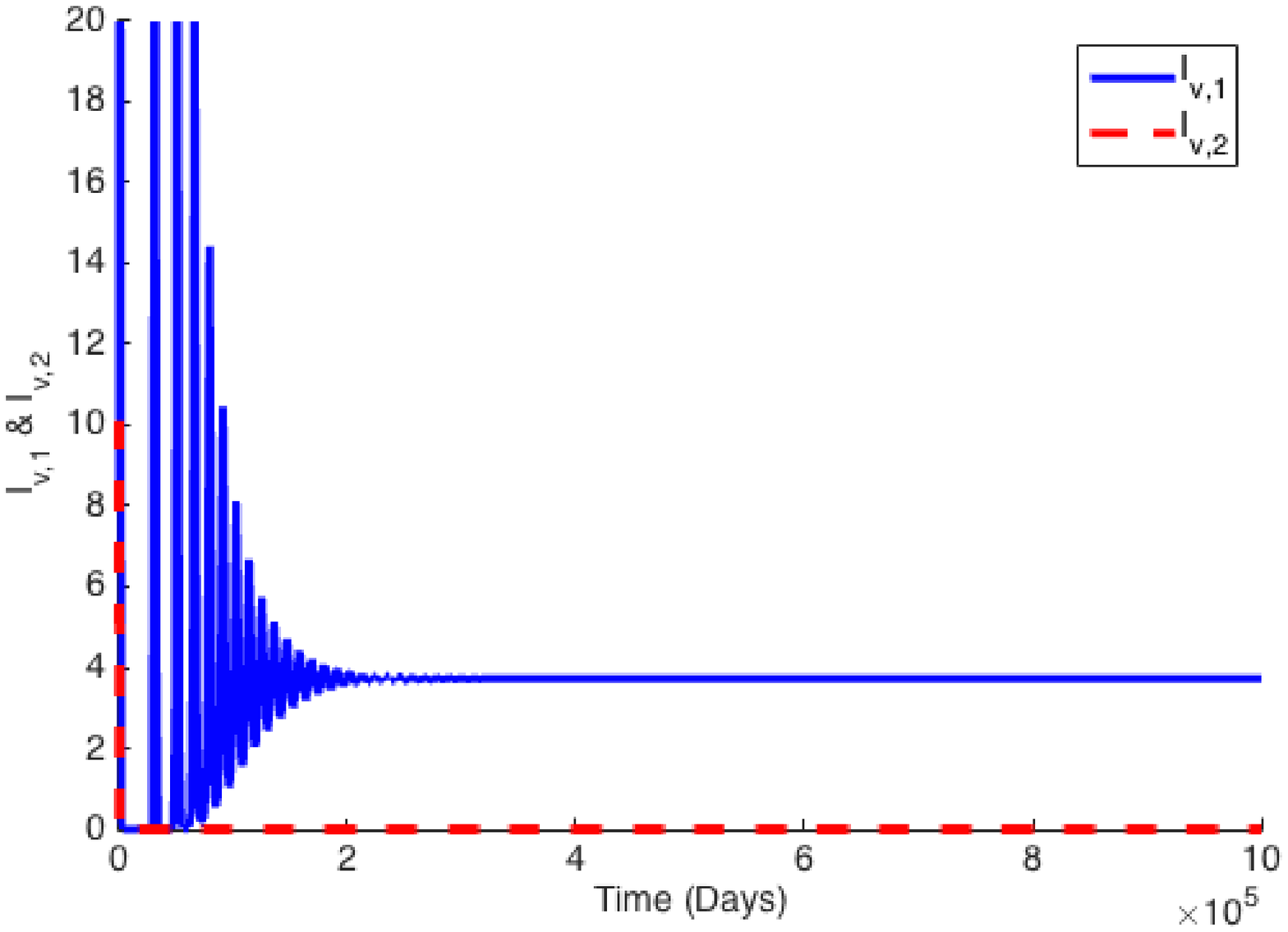}
\label{Iv1Iv2Reducible}}
\caption{Dynamics of host and vectors if the host-vector configuration matrix is reducible.} \label{fig:twofigsIhReducible}
\end{figure}

\section{Colima City and Manzanillo Dengue Inspired Simulation Study}
\label{sec:casestudy}

 \textit{Ae. aegypti} was declared eradicated in  Mexico in 1963. Not surprisingly, all four dengue serotypes (DENV-1, DENV-2, DENV-3 and DENV-4) re-emerged two years after local  the 1963 eradication \cite{diaz2006dengue}. Further, DHF cases have steadily increased since 1994 \cite{navarrete2005clinical}. Dengue is endemic in Mexico with approximately 60\% of  year-round cases reported in southern part of the country; a region is characterized by a warm and humid climate \cite{colon2011climate}. Colima, located on the central Pacific Coast (see Figure~\ref{fig:colima}), is also a reservoir of Dengue. In 2002, the State of Colima  reported 4,040  cases dengue in all  of its 10 municipalities;  495  progressing to DHF \cite{chowell2007clinical,espinoza2005clinical}. DENV-2 was isolated from patients during this outbreak  \cite{espinoza2005clinical}. The increase in DHF cases in Mexico has been linked to the introduction of   DENV-2 Asian, previously isolated in 2000 and again in 2002 \cite{lorono2004introduction}. \\

The dynamics of dengue are explored in the context of this 2002 State of Colima outbreak. The first reported (index) case was identified as that of a 10-year-old female in the municipality of Manzanillo on January 11, 2002. Dengue infection spread throughout the whole state with the most affected municipalities being Colima city, the capital of the state, and Manzanillo, an important tourist destination in the coast \cite{espinoza2005clinical}. The city of Colima reported approximately 1,167 dengue cases, with 169 cases progressing to DHF while Manzanillo, reported 1,334 dengue cases, with 123 progressing to DHF in 2002 \cite{chowell_estimation_2007}. The city of Colima and Manzanillo are  linked via high levels of travel and tourism. Both cities account for approximately 47\% of the state population. We apply a two-patch model to explore the role that movement, modeled via the matrix $p_{ij}$, may have had on dengue disease transmission during this 2002 outbreak. The estimated population of Manzanillo and Colima City were $N_{h,1}=1,355$ and $N_{h,2}=1,184$, respectively, and the initial mosquito populations were choosen to best fit the data. They were approximately 308 and 738 in Manzanillo and Colima City, respectively. Note that the host population is not the actual population of the cities but rather the population at risk in each of the corresponding cities. The population at risk is  much smaller that the actual population because in the same city there are social groups practically disconnected to others by geographic, cultural and social factors.  Entomological parameters were estimated using \cite{yang_follow_2011} and taking into account the mean temperature in each region \cite{chowell_estimation_2007}. The remaining parameters used to study the outbreak in Colima, Mexico were obtained from the literature \cite{adams_how_2010, chitnis2013modelling,garcia-rivera_dengue_2006,yang_follow_2011}:

$$
\beta_{hv}a_1 = 0.43\ \textrm{ days}^{-1},\ \beta_{hv}a_{2}=0.34\  \textrm{days}^{-1},\;  \mu_v=0.036\  \textrm{(Colima)},\ 0.030\ \textrm{(Manzanillo)}\ \textrm{days}^{-1},\;
$$
$$
\frac{1}{\mu_h}= 60 \times 365 \textrm{ days},\; \gamma_1= 0.2 \textrm{ days}^{-1},\; \gamma_2= 0.2 \textrm{ days}^{-1},\; \nu_h = 0.18 \textrm{ days}^{-1},\; \nu_v = 0.1 \textrm{ days}^{-1}.
$$

\begin{figure}[H]
\centering
\includegraphics[scale=0.35]{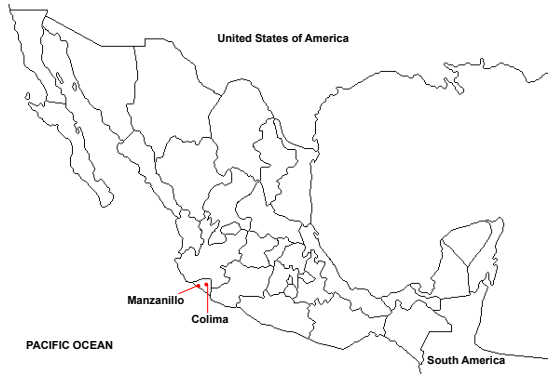}
\caption{The state of Colima is located on the central Pacific coast of Mexico. It has a tropical climate, a surface of 5,455 $km^{2}$, and a population of approximately 488,028 inhabitants. The state of Colima is divided in 10 municipalities. Manzanillo, where the 2002 outbreak began, and Colima City are labeled in the map. $p_{11}=0.99$, $p_{22}=1.0$ with Manzanillo being represented with Patch 1 and Colima with Patch 2. }
\label{fig:colima}
\end{figure}

In order to assess, within our staged scenarios, the impact of migration during the 2002 dengue outbreak, we fit the two-patch model using the incidence data for Manzanillo and the city of Colima reported by the Mexican Social Security Institute (IMSS) during the outbreak (see Figure~\ref{fig:prevdata}). The data fitting for cumulative dengue cases given by the model using `scipy.optimize.curve\_fit' library of python v2.7 programming language, is shown in Figure~\ref{fit}. Model results show that dengue spreads more quickly in the city of Colima when the proportion of visits from Manzanillo's infected residents is high, see the left panel of Figure \ref{fig:experiments} compared with Figure \ref{fit}. Alternatively, susceptible Colima City residents would acquire dengue infections over a longer time frame in Manzanillo, introducing the disease over a slower time scale in  their home residence, the city of Colima. Of course, the absence of movement  leads to no dengue cases in Manzanillo; an outbreak occurring only in Colima $p_{11}=p_{22}=1.0$, see center panel of Figure \ref{fig:experiments}; equal movement, $p_{11}=p_{22}$, would cause the outbreak in Colima to grow faster, as can be seen in the right panel of Figure \ref{fig:experiments}.  Hence, limiting the movement of the Manzanillo population seems like a good strategy while limiting the movement of the Colima population wouldn't be as effective. In the latest scenario, the economic cost would be high since Manzanillo is a tourist destination.

\begin{figure}[H]
\centering
\includegraphics[scale=0.35]{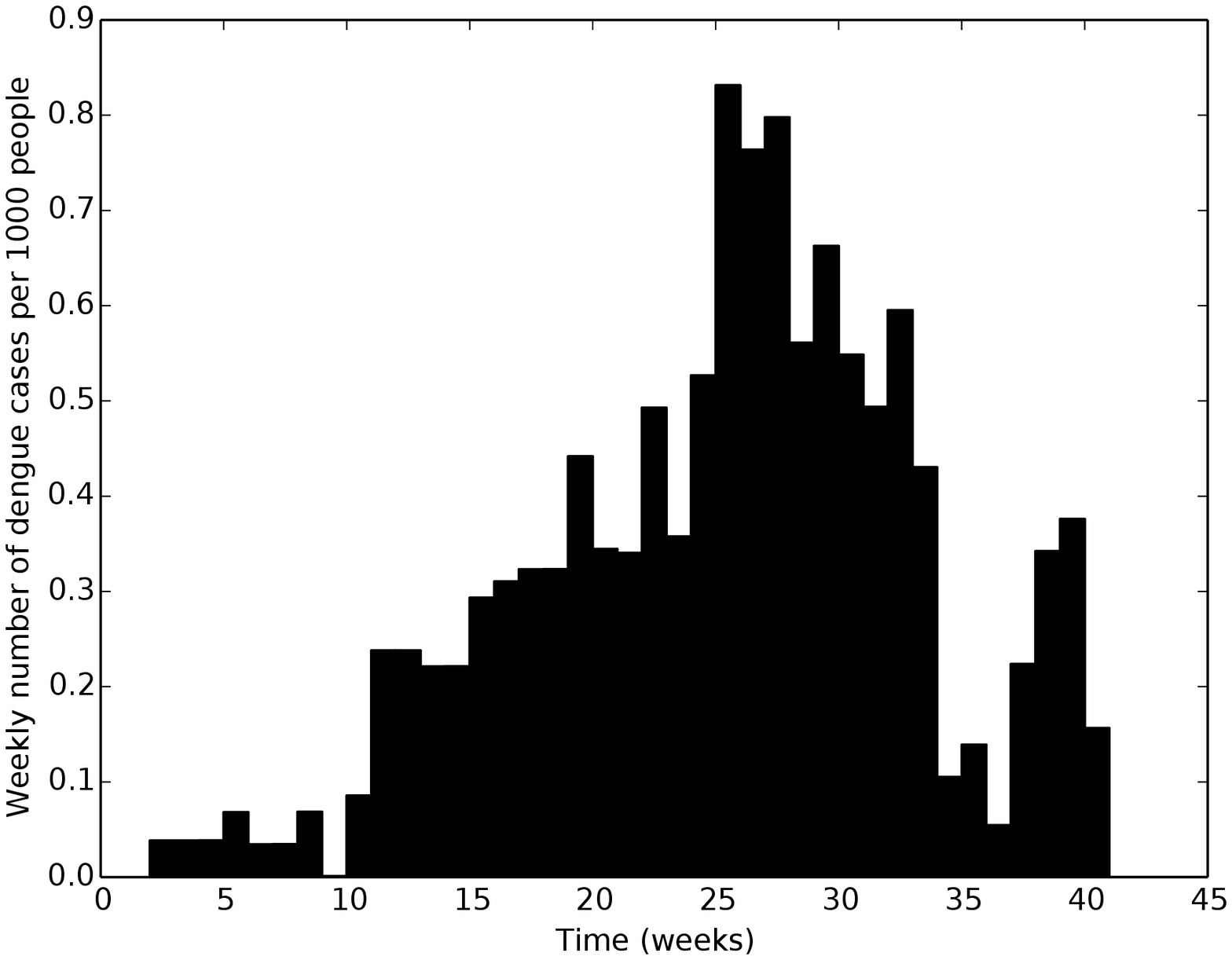}
\includegraphics[scale=0.35]{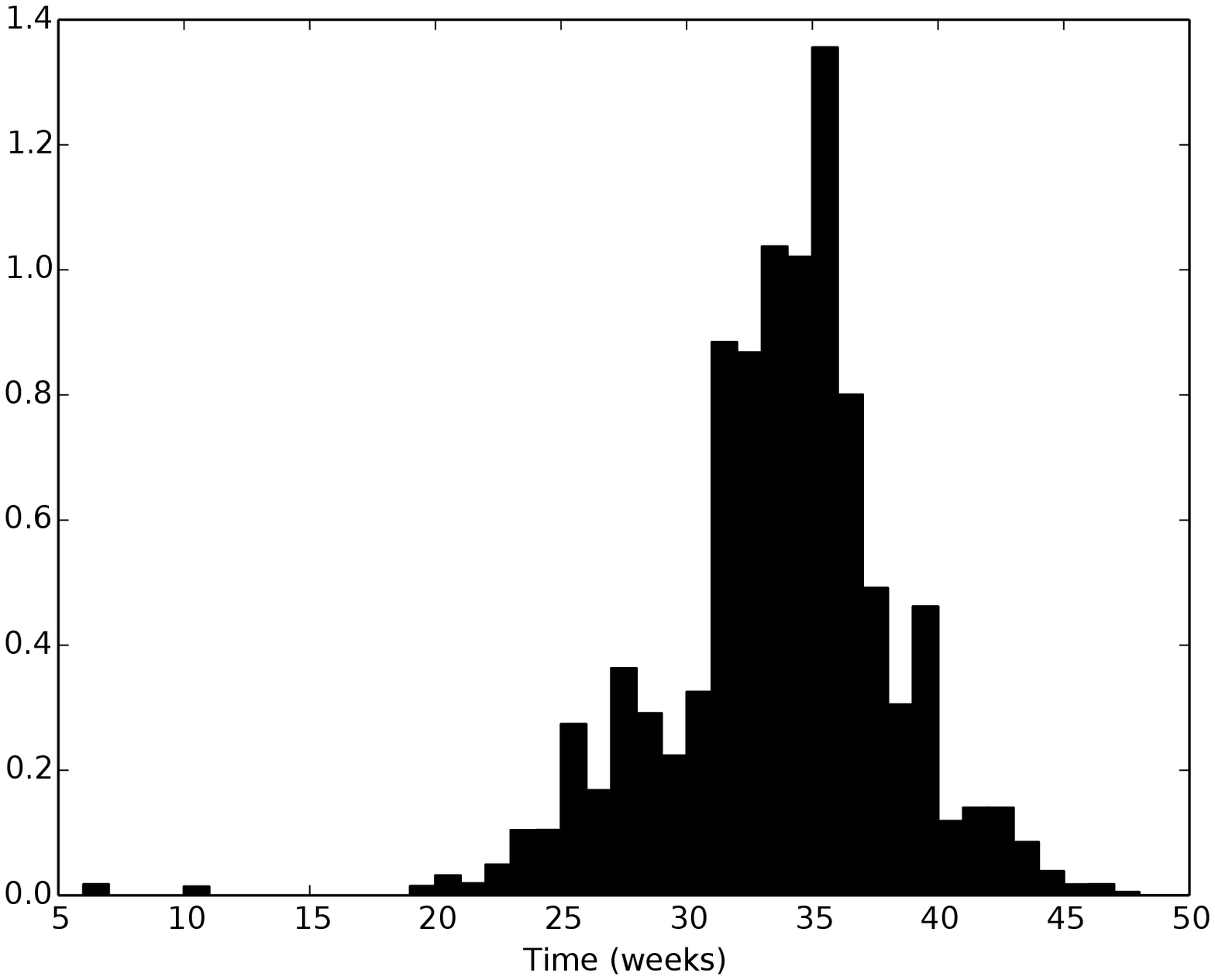}
\caption{ Incidence of dengue cases per weekly during the 2002 dengue epidemic diagnosed at the hospitals of the Mexican Institute of Public Health (IMSS) \cite{chowell_estimation_2007} in Manzanillo (left) and Colima city (right), respectively.}
\label{fig:prevdata}
\end{figure}

\begin{figure}[H]
\centering
\includegraphics[width=0.33\textwidth]{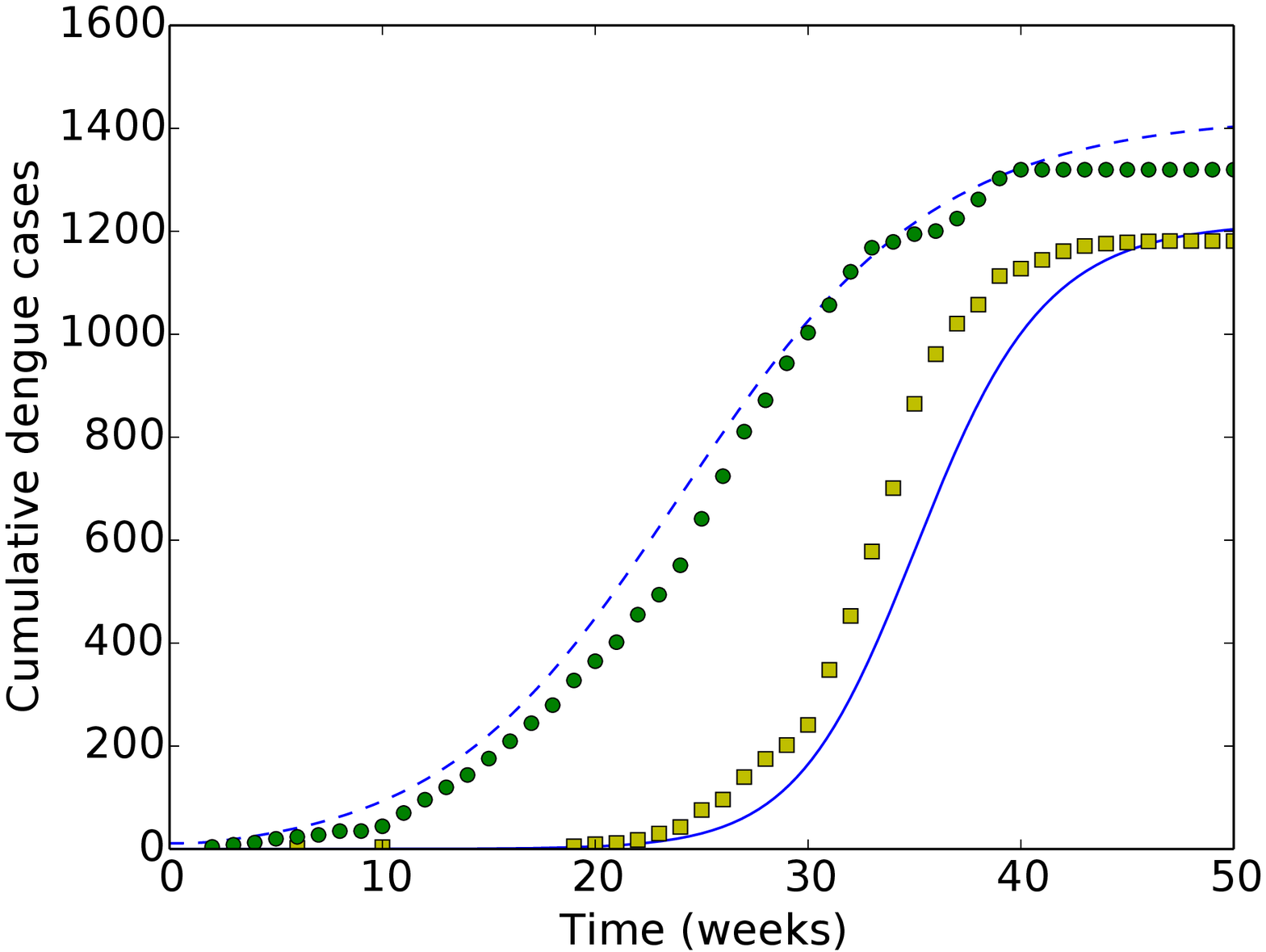}%
\includegraphics[width=0.33\textwidth]{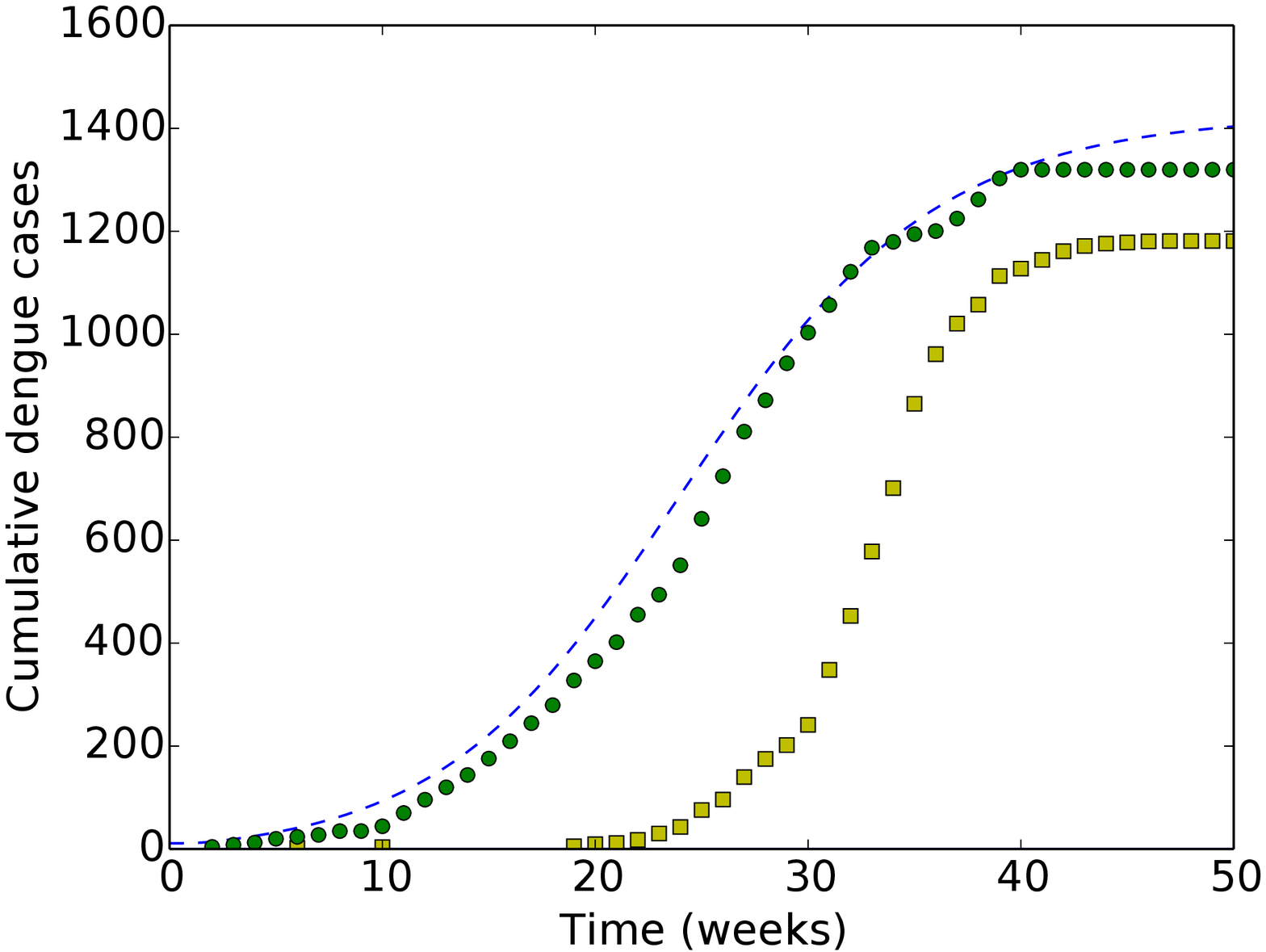}%
\includegraphics[width=0.33\textwidth]{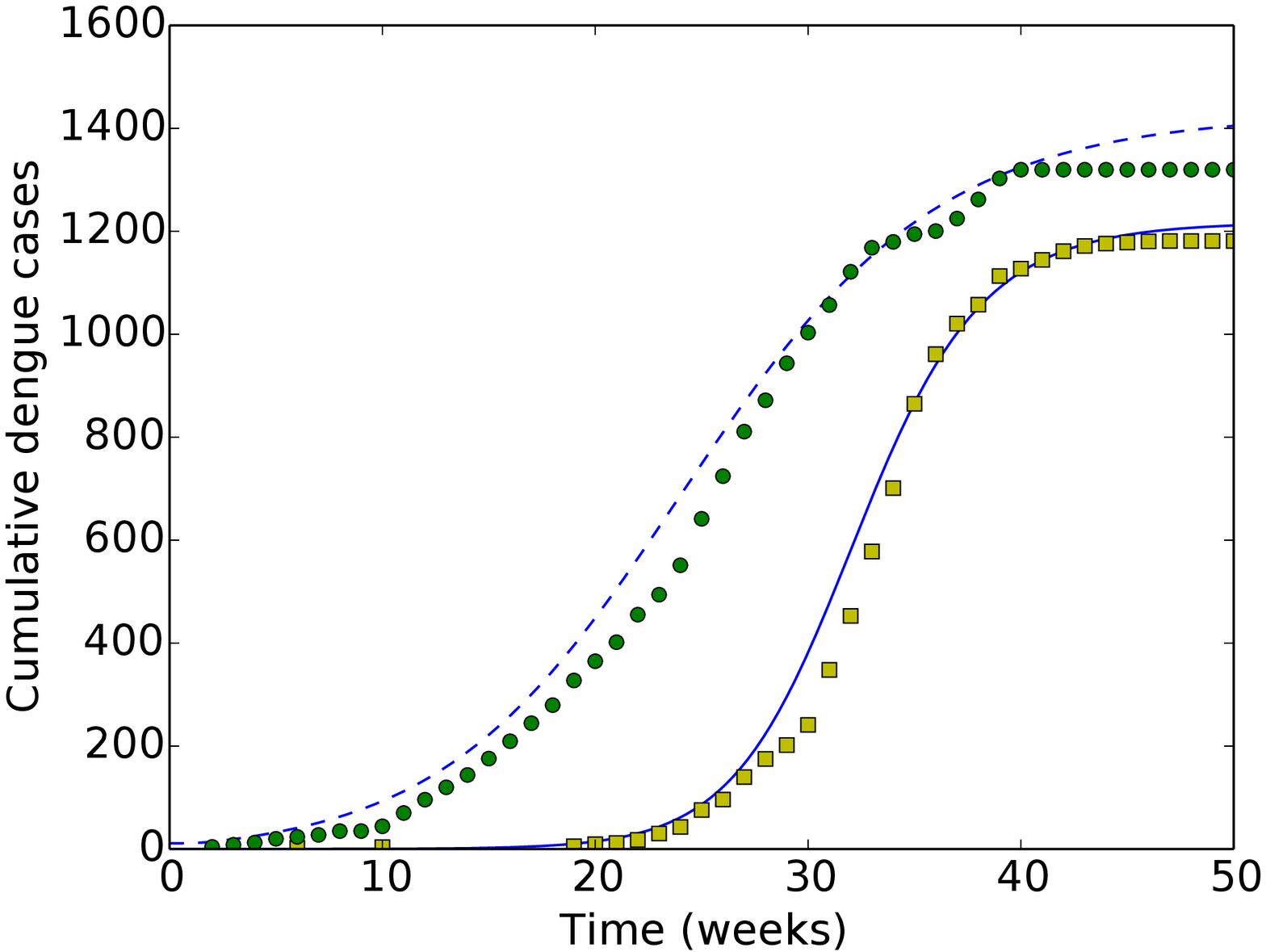}%
\caption{Circles: Cumulative dengue cases reported on Manzanillo, dotted line: Model prediction of Manzanillo cases, squares: Cumulative dengue cases reported on Colima city, solid line: Model prediction of cumulative cases in Colima city. Left: $p_{11}=1.0$, $p_{22}=0.9996$, Center: $p_{11}=1.0$, $p_{22}=1.0$, Right: $p_{11}=p_{22}=0.9996$ Patch 1 represents Manzanillo and patch 2 Colima. }
\label{fig:experiments}
\end{figure}

\begin{figure}[H]
\begin{center}
\includegraphics[scale=0.4]{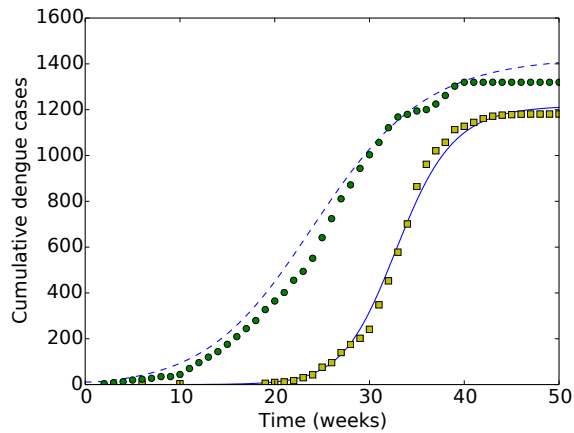}
\caption{Best fit of the model $p_{11}=0.9996$, $p_{22}=1.0$. Circles: Cumulative dengue cases reported on Manzanillo, dotted line: Model prediction of Manzanillo cases, squares: Cumulative dengue cases reported on Colima city, solid line: Model prediction of cumulative cases in Colima city. }
\label{fit}
\end{center}
\end{figure}
%
We can also observe in Figure~\ref{Hipotetical} (on the left), that the effect of reducing the transit from Manzanillo to Colima city led only to a delay in the appearance of the outbreak in Colima. This indicates that the outbreak in Colima followed its own local dynamics and that transit between these two cities only led to delays in the introduction of the dengue virus without affecting the local outbreak dynamics.  When the average visiting time spent in a place where the disease prevalence is low (small value of $p_{ij}$, $i\neq j$) then the only way of reducing an outbreak would require strict migration control, that is,   complete  travel avoidance to the high risk zone. In Figure~\ref{Hipotetical} (on the right), we see that with only a small fraction of visitors from Manzanillo to Colima, the outbreaks in both cities occur almost simultaneously. Model simulations re-affirm the views that the rate of host movement and time spent in endemic geographic regions are important for the spread of dengue between two patches. The question then becomes, why aren't then these residence times estimated?

\begin{figure}[H]
\includegraphics[scale=0.4]{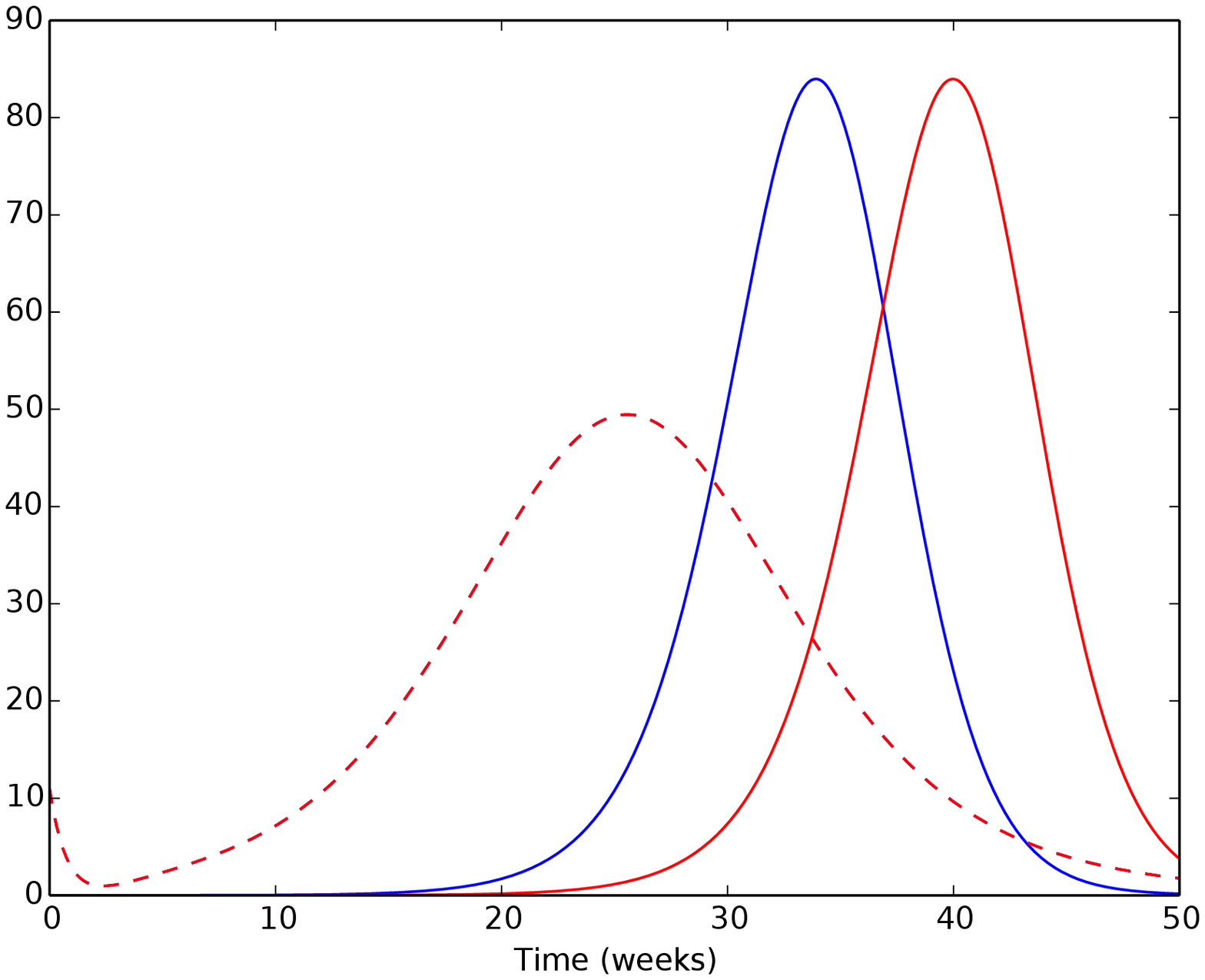}
\includegraphics[scale=0.4]{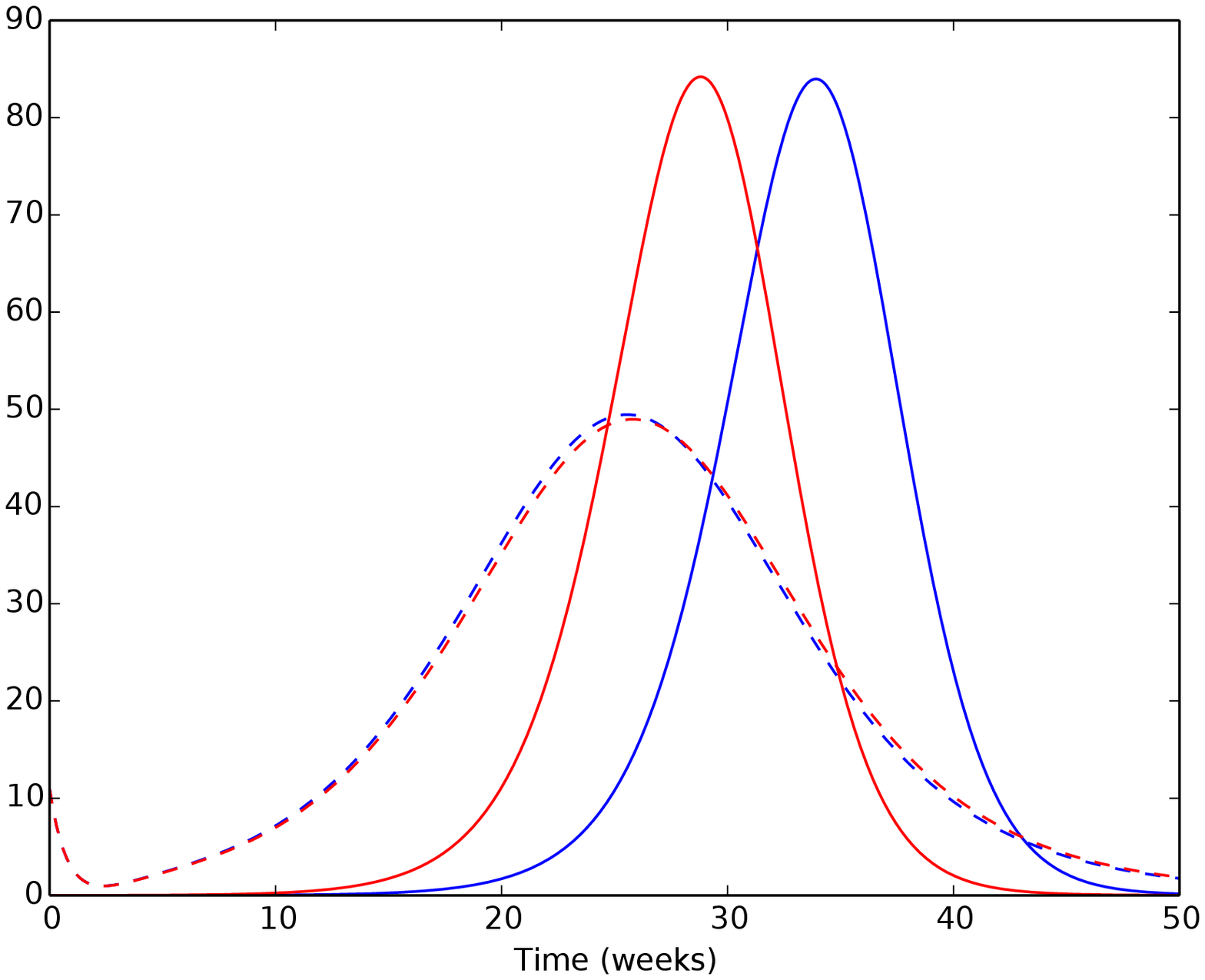}
\caption{Dengue cases predicted by the model for Manzanillo (dotted line) and Colima city (solid line). In the first scenario (on the left), the blue lines represent no transit control and the red lines represent a reduction of 90\% in movement from Manzanillo city to Colima. In the second scenario (on the right), the blue lines represent no movement control and the red lines represent an increment of movement from Colima to Manzanillo city of 1\%.}
\label{Hipotetical}
\end{figure}

\section{Conclusion}
\label{sec:conclusion}
The persistence of vector-bone diseases, such as dengue, is connected to factors that include the presence ecological conditions  that favor high vector densities, vector-host interactions,  the spatial movement of humans, and of course, the effectiveness of control measures \cite{martens2000malaria,sutherst2004global}. In this paper, a two-patch host-vector model was used to study the role of movement on the transmission dynamics of dengue, especially DENV-2.  We focus on the applications of our framework to scenarios where  dengue is endemic and where vertical transmission has been documented. A residence times matrix $\mathbb P$ is used to model host mobility. This modeling approach provides a framework for exploring spatial vector-borne disease dynamics and control within relatively ``close' environments. Analytical results were derived and the conditions for which the disease dies out or persists have been identified;  conditions that depend on whether the basic reproduction number $\mathcal R_0(\mathbb P)$  is less or greater than unity and the connectivity of  patches.\\ 

 Using data from the 2002 DENV-2 outbreak in Colima, Mexico, we compare the overall prevalence in the cities of Colima and Manzanillo as a function of pre-selected $\mathbb P$ matrices. Our model shows that reducing traveling from  to Colima city, considered high-risk and the place of the 2002 outbreak onset, causes a slight delay in the spread of the disease. In order to completely prevent an outbreak in Colima city, migration between Colima city and Manzanillo must be stopped. Manzanillo a tourist destination implies that transit from Colima city to Manzanillo is expected to peak during certain seasons. The model suggests that dengue would become endemic in both patches almost simultaneously. The two-patch model highlights the role of human spatial movement on disease transmission and control. The strength of this effect depends on the proportion of time commuters to high or low risk spend in each patch.


\appendix
\section{Appendix: The basic reproduction number}\label{R0}
Let $x=(E_{h,1},E_{h,2},E_{v,1},E_{v,2},I_{h,1},I_{h,2},I_{v,1},I_{v,2})$ and so the relevant $\mathcal F$ and $\mathcal V$ are
$$\mathcal F=\left(\begin{array}{c}
\frac{a_1\beta_{vh}p_{11}S_{h,1}I_{v,1}}{p_{11}N_{h,1}+p_{21}N_{h,2}}+\frac{a_2\beta_{vh}p_{12}S_{h,1}I_{v,2}}{p_{12}N_{h,1}+p_{22}N_{h,2}}\\
\frac{a_1\beta_{vh}p_{21}S_{h,2}I_{v,1}}{p_{11}N_{h,1}+p_{21}N_{h,2}}+\frac{a_2\beta_{vh}p_{22}S_{h,2}I_{v,2}}{p_{12}N_{h,1}+p_{22}N_{h,2}}\\
a_1\beta_{hv}S_{v,1}\frac{ p_{11}I_{h,1}+p_{21}I_{h,2}}{ p_{11}N_{h,1}+p_{21}N_{h,2}}\\
a_2\beta_{hv}S_{v,2}\frac{ p_{12}I_{h,1}+p_{22}I_{h,2}}{ p_{12}N_{h,1}+p_{22}N_{h,2}}\\
0\\
0\\
0\\
0
\end{array}\right)\quad\text{and}\quad \mathcal V=\left(\begin{array}{c}
-(\mu_h+\nu_h) E_{h,1}\\
-(\mu_h+\nu_h) E_{h,2}\\
-(\mu_v+\nu_v)E_{v,1}\\
-(\mu_v+\nu_v)E_{v,2}\\
\nu_hE_{h,1}-(\mu_h+\gamma_i) I_{h,1}\\
\nu_hE_{h,2}-(\mu_h+\gamma_i) I_{h,2}\\
\nu_vE_{h,1}-(1-q)\mu_vI_{v,1}\\
\nu_vE_{h,2}-(1-q)\mu_vI_{v,2}
\end{array}\right).
$$
Let  $F\equiv D\mathcal F$ and $V\equiv D\mathcal V$ evaluated at the DFE. We obtain,
{\tiny{
$$
F=\left(\begin{array}{ccccc}
0_{2,4} & \begin{array}{cc}0 & 0\\ 0 & 0\end{array}  &\begin{array}{cc} \frac{a_1\beta_{vh}p_{11}N_{h,1}}{ (1-p)\mu_v(p_{11}N_{h,1}+p_{21}N_{h,2})} & \frac{a_2\beta_{vh}p_{12}N_{h,1}}{(1-p)\mu_v( p_{12}N_{h,1}+p_{22}N_{h,2})}\\\frac{a_1\beta_{vh}p_{21}N_{h,2}}{ (1-p)\mu_v(p_{11}N_{h,1}+p_{21}N_{h,2})} & \frac{a_2\beta_{vh}p_{22}N_{h,2}}{ (1-p)\mu_v(p_{12}N_{h,1}+p_{22}N_{h,2}}\end{array}\\
0_{4,4} &  \begin{array}{cc}\frac{a_1\beta_{vh}p_{11}N_{v,1}}{ (\mu_h+\gamma_1)(p_{11}N_{h,1}+p_{21}N_{h,2})} & \frac{a_1\beta_{vh}p_{21}N_{v,1}}{ (\mu_h+\gamma_2)(p_{11}N_{h,1}+p_{21}N_{h,2})}\\ \frac{a_2\beta_{vh}p_{12}N_{v,2}}{ (\mu_v)(p_{12}N_{h,1}+p_{22}N_{h,2})} & \frac{a_2\beta_{vh}p_{22}N_{v,2}}{ (\mu_h+\gamma_2)(p_{21}N_{h,1}+p_{22}N_{h,2}})\end{array} &  0_{2,2} \\
0_{4,4}  & 0_{4,2} & 0_{4,2}   \\
\end{array}\right)
$$ }} 
and
{\tiny{
$$
V=\left(\begin{array}{cccccccc}
-\mu_h-\nu_h &0 & 0 & 0 & 0  & 0  &0 & 0\\ 
0 &-\mu_h-\nu_h & 0 & 0 & 0  & 0  &0 & 0\\ 
0 &0 & -\mu_v-\nu_v & 0 & 0  & 0  &0 & 0\\ 
0 &0 & 0 & -\mu_v-\nu_v& 0  & 0  &0 & 0\\ 
\nu_v&0 & 0 & 0 & -\mu_h-\gamma_1  & 0  &0 & 0\\ 
0 &\nu_v& 0 & 0 & 0  & \mu_h-\gamma_h  &0 & 0\\  
0 &0 & \nu_v& 0 & 0  & 0  & -(1-p)\mu_v & 0\\
0 &0 & 0 & \nu_v & 0  & 0  &0 & -(1-p)\mu_v
\end{array}\right).
$$ }}

The basic reproduction number is the spectral radius of the matrix,
{\tiny{
$$
-FV^{-1}=\left(\begin{array}{cccccc}
0_{2,2} & M_{vh} & \begin{array}{cc}0 & 0\\ 0 & 0\end{array}  &\begin{array}{cc} \frac{a_1\beta_{vh}p_{11}N_{h,1}}{ (1-q)\mu_v(p_{11}N_{h,1}+p_{21}N_{h,2})} & \frac{a_2\beta_{vh}p_{12}N_{h,1}}{(1-q)\mu_v( p_{12}N_{h,1}+p_{22}N_{h,2})}\\\frac{a_1\beta_{vh}p_{21}N_{h,2}}{ (1-q)\mu_v(p_{11}N_{h,1}+p_{21}N_{h,2})} & \frac{a_2\beta_{vh}p_{22}N_{h,2}}{ (1-q)\mu_v(p_{12}N_{h,1}+p_{22}N_{h,2}}\end{array}\\
M_{hv} & 0_{2,2} &  \begin{array}{cc}\frac{a_1\beta_{vh}p_{11} N_{v,1}}{ (\mu_h+\gamma_1)(p_{11}N_{h,1}+p_{21}N_{h,2})} & \frac{a_1\beta_{vh}p_{21} N_{v,1}}{ (\mu_h+\gamma_2)(p_{11}N_{h,1}+p_{21}N_{h,2})}\\ \frac{a_2\beta_{vh}p_{12}N_{v,2}}{ (\mu_v)(p_{12}N_{h,1}+p_{22}N_{h,2})} & \frac{a_2\beta_{vh}p_{22}N_{v,2}}{ (\mu_h+\gamma_2)(p_{21}N_{h,1}+p_{22}N_{h,2}})\end{array} &  0_{2,2} \\
0_{4,2}  & 0_{4,2} & 0_{4,2}  & 0_{4,2} \\
\end{array}\right)
$$ }} 
where 
$$
   M_{vh}= \left(    \begin{array}{cc}\frac{a_1\beta_{vh}p_{11} N_{h,1}\nu_v}{(p_{11}N_{h,1}+p_{21}N_{h,2})(\mu_v+\nu_v)(1-q)\mu_v} & \frac{a_2\beta_{vh}p_{12}N_{h,1}\nu_v}{(p_{12}N_{h,1}+p_{22}N_{h,2}) (\mu_v+\nu_v)(1-q)\mu_v}\\ \frac{a_1\beta_{vh}p_{21}N_{h,2}\nu_v}{ (p_{11}N_{h,1}+p_{21}N_{h,2})(\mu_v+\nu_v)(1-q)\mu_v} & \frac{a_2\beta_{vh}p_{22}N_{h,2}\nu_v}{ (p_{12}N_{h,1}+p_{22}N_{h,2})(\mu_v+\nu_v)(1-q)\mu_v}\end{array}\right)
$$
and 
$$
   M_{hv}= \left(    \begin{array}{cc}\frac{a_1\beta_{hv}p_{11}  N_{v,1}\nu_h}{(p_{11}N_{h,1}+p_{21}N_{h,2})(\mu_h+\nu_h)(\mu_h+\gamma_1)} & \frac{a_1\beta_{hv}p_{21} N_{v,1}\nu_h}{(p_{11}N_{h,1}+p_{21}N_{h,2}) (\mu_h+\nu_h)(\mu_h+\gamma_2)}\\
   \frac{a_2\beta_{hv}p_{12} N_{v,2}\nu_h}{ (p_{12}N_{h,1}+p_{22}N_{h,2})(\mu_h+\nu_h)(\mu_h+\gamma_1)} & \frac{a_2\beta_{hv}p_{22} N_{v,2}\nu_h}{ (p_{12}N_{h,1}+p_{22}N_{h,2})(\mu_h+\nu_h)(\mu_h+\gamma_2)}
   \end{array}\right).
$$
The basic reproduction number $\mathcal R_0^2$ is defined by the expression, $$\mathcal R_0^2=\rho(M_{vh}M_{hv}).$$
\end{document}